\documentclass[11pt]{article}

\usepackage[english]{babel}
\usepackage[utf8x]{inputenc}
\usepackage[T1]{fontenc}
\usepackage{amsmath,amsxtra,amssymb,amsfonts,amsthm,amscd,array}
\usepackage{graphicx,latexsym,wrapfig}

\DeclareMathOperator*{\argmin}{arg\,min}
\usepackage{color}
\usepackage{mathtools}
\usepackage{float}
\usepackage{booktabs}
\usepackage{fancyhdr}
\usepackage{lscape}
\usepackage[colorinlistoftodos]{todonotes}
\usepackage[colorlinks=true, allcolors=blue]{hyperref}
\usepackage[authoryear,sort]{natbib}
\usepackage{appendix}
\usepackage{enumitem}
\usepackage{tikz}
\usetikzlibrary{matrix,chains,positioning,decorations.pathreplacing,arrows}
\usepackage{makecell}

\newcommand{\crl}[1]{\ensuremath{ \left\{ #1 \right\} }}
\newcommand{\edg}[1]{\ensuremath{\! \left[ #1 \right] }}
\newcommand{\brak}[1]{\ensuremath{\left( #1 \right)}}

\newcommand{\n}[1]{\ensuremath{ \| #1 \| }}
\newcommand{\no}[1]{\ensuremath{ \left\| #1 \right\| }}
\newcommand{\abs}[1]{\ensuremath{ \left| #1 \right| }}

\counterwithin*{equation}{section}
 
 \newcommand{\be}{\begin{equation}}
 \newcommand{\ee}{\end{equation}}
 \newcommand{\bea}{\begin{eqnarray}}
 \newcommand{\eea}{\end{eqnarray}}
 \newcommand{\beas}{\begin{eqnarray*}}
\newcommand{\eeas}{\end{eqnarray*}}

\usepackage{xparse}
\usepackage{siunitx}  
\usepackage{colortbl} 

\newcommand{\classname}[1]{\texttt{#1}} 

\newtheorem{theorem}{Theorem}[section]

\newtheorem{proposition}[theorem]{Proposition}

\newtheorem{lemma}[theorem]{Lemma}
\newtheorem{remark}[theorem]{Remark}
\newtheorem{example}[theorem]{Example}
\newtheorem{examples}[theorem]{Examples}
\newtheorem{framework}[theorem]{Framework}
\newtheorem{foo}[theorem]{Remarks}

\newenvironment{Remark}{\begin{remark}\rm}{\end{remark}}

\usepackage{style}

\newcommand{\R}{\mathbb{R}}

\newcommand{\p}{\mathbb{P}}

\newcommand{\E}{\mathbb{E}}
\newcommand{\Var}{\operatorname{Var}}

\newcommand{\cF}{\mathcal{F}} 
\newcommand{\cS}{\mathcal{S}}

\newcommand{\myrowcolour}{\rowcolor[gray]{0.925}}

\title{Computation of conditional expectations\\ with guarantees\thanks{We thank 
Daniel Bartl, Sebastian Becker and Peter B\"uhlmann for fruitful discussions and helpful comments.}}
\author{
  Patrick Cheridito \\  
  Department of Mathematics\\
  ETH Zurich, Switzerland 
   \And
 Balint Gersey \\
  Department of Mathematics\\
  ETH Zurich, Switzerland
}
\begin{document}
\maketitle

\begin{abstract}
Theoretically, the conditional expectation of a square-integrable random variable 
$Y$ given a $d$-dimensional random vector $X$ can be obtained by minimizing the 
mean squared distance between $Y$ and $f(X)$ over all Borel measurable functions 
$f \colon \R^d \to \R$. However, in many applications this minimization problem cannot be solved exactly,
and instead, a numerical method which computes an approximate minimum 
over a suitable subfamily of Borel functions has to be used. The quality of the result depends on the adequacy of the subfamily and the performance of the numerical method. In this paper, we derive an 
expected value representation of the minimal mean squared distance which in many  
applications can efficiently be approximated with a standard Monte Carlo average.
This enables us to provide guarantees for the accuracy of any numerical approximation 
of a given conditional expectation. We illustrate the method by assessing the quality
of approximate conditional expectations obtained by linear, polynomial and neural 
network regression in different concrete examples.
\end{abstract}

$\mbox{}$\\[-1.5mm]

\section{Introduction}
\label{sec:intro}

The goal of this paper is to compute the conditional expectation $\E[Y \mid X]$ of a square-integrable 
random variable $Y$ given a $d$-dimensional random vector $X$, both 
defined on a common probability space $(\Omega, \cF, \p)$. The accurate estimation of 
conditional expectations is an important problem arising in different branches of science and
engineering as well as finance, economics and various business applications. 
In particular, it plays a central role in 
regression analysis, which tries to model the relationship between a 
response variable $Y$ and a number of explanatory variables $X_1, \dots, X_d$
\citep[see, e.g.,][]{norman98, ryan2008modern, hastie2009elements, chatterjee2015regression}. 
But it also appears in different computational problems, such as the numerical approximation 
of partial differential equations and backward stochastic differential equations
\citep[see, e.g.,][]{bally1997approximation, chevance1997numerical, BT, gobet2005regression, GobetT, FTW, beck2019deep},
stochastic partial differential equations \citep[see, e.g.,][]{beck2020deep},
stochastic control problems \citep[see, e.g.,][]{astrom70,bain2008fundamentals},
stochastic filtering \citep[see, e.g.,][]{jazwinski2007stochastic}, 
complex financial valuation problems \citep[see, e.g.][]{carriere96,longstaff2001valuing, tsitsiklis2001regression,broadie2004stochastic, broadie2008,becker2020pricing}
as well as financial risk management \citep[see, e.g.][]{lee2003computing, gordy2010nested, 
broadie2011efficient, BauerReussSinger, Cher}.
In addition, conditional expectations are closely related to 
squared loss minimization problems arising in various machine learning 
applications \citep[see, e.g.,][]{hastie2009elements, goodfellow2016deep}.

If it is possible to simulate from the conditional distribution of $Y$ given $X$, 
the conditional expectation $\E[Y \mid X]$ can be approximated with nested Monte Carlo simulation; see, e.g., 
\citep[see, e.g.,][]{BauerReussSinger, broadie2011efficient, broadie2015risk}. 
While the approach can be shown to converge for increasing sample sizes, 
it often is too time-consuming to be useful in practical applications.
On the other hand, it is well known that $\E[Y \mid X]$ is of the form $\bar{f}(X)$ 
for a regression function $\bar{f} \colon \R^d \to \R$ which can be characterized as
a minimizer\footnote{The conditional expectation $\E[Y \mid X]$ is unique up to $\p$-almost sure equality.
Accordingly, the regression function $\bar{f}$ is unique up to almost sure equality 
with respect to the distribution of $X$.} 
of the mean squared distance
\be \label{msd}
\mathbb{E} \edg{(Y - f(X))^2}
\ee
over all Borel functions $f \colon \R^d \to \R$ \citep[see, e.g.,][]{bru1985meilleures}. 
However, in many applications, the minimization problem \eqref{msd} cannot be solved exactly. For instance,
the joint distribution of $X$ and $Y$ might not be known precisely, or the 
problem might be too complicated to admit a closed-form solution.
In such cases, it can be approximated with a least squares regression, 
consisting in minimizing an empirical mean squared distance
\be \label{emsd}
\frac{1}{M} \sum_{m=1}^M \brak{\mathcal{Y}^m - f(\mathcal{X}^m)}^2 
\ee
based on realizations $(\mathcal{X}^m ,\mathcal{Y}^m)$ of $(X,Y)$ over a suitable 
family ${\cal S}$ of Borel functions $f \colon \R^d \to \R$. 
This typically entails the following three types of approximation errors: 
\begin{itemize}
\item[{\rm (i)}] a function approximation error if the true regression function $\bar{f}$ 
does not belong to the function family ${\cal S}$; 

\item[{\rm (ii)}] 
a statistical error stemming from estimating the expected value \eqref{msd} with 
\eqref{emsd};

\item[{\rm (iii)}] a numerical error if the minimization of \eqref{emsd} over $\cS$ has to be solved numerically. 
\end{itemize}
Instead of analyzing the errors (i)--(iii), we here derive an alternative 
representation of the minimal mean squared distance $\E [(Y - \bar{f}(X))^2]$, 
which does not involve a minimization problem or require knowledge of the 
true regression function $\bar{f}$. This enables us to provide quantitative estimates on 
the accuracy of any numerical approximation $\hat{f}$ of $\bar{f}$. In particular, if $\hat{f}$ 
is determined with a machine learning method that is difficult to interpret, our approach contributes to 
trustworthy AI.

While the empirical mean squared distance \eqref{emsd} can directly be minimized 
using realizations $(\mathcal{X}^m, \mathcal{Y}^m)$ of $(X,Y)$, 
our approach to derive error bounds for the approximation of $\bar{f}$ requires $Y$ to
be of the form $Y = h(X,V)$ for a known function $h \colon \R^{d+k} \to \R$ and a $k$-dimensional random 
vector $V$ independent of $X$. In typical statistical applications, only realizations of $(X,Y)$ can be observed 
and a structure of the form $Y = h(X,V)$ would have to be inferred from the data. 
But in many of the computational problems mentioned above, $Y$ is directly given in the 
form $Y = h(X,V)$.

The rest of the paper is organized as follows: In Section \ref{sec:appr}, we first 
introduce the notation and some preliminary results before
we formulate the precise mean squared distance minimization problem we are considering
along with its empirical counterpart. Then we discuss upper bounds of the minimal 
mean squared distance and their approximation with Monte Carlo averages. 
In Section \ref{sec:est} we derive an expected value representation of the 
minimal mean squared distance which makes it possible to derive bounds on 
the $L^2$-error of any numerical approximation $\hat{f}$ of the true regression function $\bar{f}$.
In Section \ref{sec:ex} we compute conditional expectations in different examples 
using linear regression, polynomial regression and feedforward neural networks with varying 
activation functions. We benchmark the numerical results against values obtained from 
our expected value representation of the 
minimal mean squared distance and derive $L^2$-error estimates.
Section \ref{sec:conclusion} concludes, and in the Appendix we report 
auxiliary numerical results used to compute the figures shown in 
Section \ref{sec:ex}.

\section{Numerical approximation of conditional expectations}
\label{sec:appr}

\subsection{Notation and preliminaries}
\label{sec:notation}

Let us first note that the mean squared distance \eqref{msd} does not 
necessarily have to be minimized with respect to the original probability measure $\p$.
Indeed, the regression function $\bar{f} \colon \R^d \to \R$ only 
depends on the conditional distribution of $Y$ given $X$ and not on the distribution 
$\nu_X$ of $X$. More precisely, the measure $\p$ can be disintegrated as 
\[
\p[A] = \int_{\R^d} \p[A \mid X = x] d\nu_X(x), \quad A \in {\cal F},
\]
where $\p[. \mid X = x]$ is a regular conditional version of $\p$ given $X$. For any Borel 
probability measure $\nu$ on $\R^d$ that is absolutely 
continuous with respect to $\nu_X$, 
\[
\p^{\nu}[A] := \int_{\R^d} \p[A \mid X = x] d\nu(x), \quad A \in {\cal F},
\]
defines a probability measure on $\Omega$ under which $X$ has the modified 
distribution $\nu$ while the conditional distribution of $Y$ given $X$ is the same as 
under $\p$. Let us denote by $\E^{\nu}$ the expectation with respect to $\p^{\nu}$ 
and by ${\cal B}(\R^d ; \R)$ the set of all Borel functions $f \colon \R^d \to \R$. 
With this notation, one has the following.

\begin{lemma} \label{lemma:choice}
Assume $\E^{\nu} Y^2 < \infty$. Then a minimizer $\tilde{f} \colon \R^d \to \R$ of the 
distorted minimal mean squared distance
\be \label{D}
D^{\nu} := \min_{f \in {\cal B}(\R^d ; \, \R)} \E^{\nu} \edg{ \brak{Y - f(X)}^2 }
\ee
agrees with $\bar{f} \colon \R^d \to \R$ $\nu$-almost surely. 
In particular, if $\nu$ has the same null sets as $\nu_X$, then $\tilde{f} = \bar{f}$ $\nu_X$-almost surely.
\end{lemma}

\begin{proof}
A Borel function $\tilde{f} \colon \R^d \to \R$ minimizes \eqref{D} if and only if 
\[
\tilde{f}(x) = \argmin_{z \in \R} \int_{\R} (y - z)^2 \, \p[Y \in dy \mid X = x] \quad \mbox{for $\nu$-almost all } x \in \R^d.
\]
Since $\bar{f}$ has an analogous representation holding for $\nu_X$-almost all $x \in \R^d$,
it follows that $\tilde{f}$ agrees with $\bar{f}$ $\nu$-almost surely.
In particular, if $\nu$ has the same null sets as $\nu_X$, then 
$\tilde{f} = \bar{f}$ $\nu_X$-almost surely.
\end{proof}

Lemma \ref{lemma:choice} gives us the flexibility to choose a distribution $\nu$ on $\R^d$ which
assigns more weight than $\nu_X$ to regions of $\R^d$ that are important in a given application.
For instance, $\nu \ll \nu_X$ can be chosen so as to concentrate more weight around 
a given point $x_0$ in the support of $\nu_X$; see Lemma \ref{lemma:point} 
and Section \ref{ex:non-polydistorted} below. On the other hand, in financial risk management
one is usually concerned with the tails of loss distributions. Then the distribution 
$\nu_X$ can be tilted in the direction of large losses of a financial exposure;
see Section \ref{ex:Bindistorted} below.

\subsection{Upper bound of the minimal mean squared distance}
\label{sec:upper}

In many situations, the minimization problem \eqref{D} cannot be solved exactly.
But if one has access to $\p^{\nu}$-realizations $(\mathcal{X}^m, \mathcal{Y}^m)$ of $(X,Y)$,
the true regression function $\bar{f}$ can be approximated by minimizing the empirical 
mean squared distance
\be \label{empD}
\frac{1}{M} \sum_{m =1}^M \brak{\mathcal{Y}^m - f(\mathcal{X}^m)}^2
\ee
over $f$ in a subset ${\cal S}$ of ${\cal B}(\R^d ; \R)$. In the examples of Section \ref{sec:ex} below, we compare 
results obtained by using linear combinations of $1, X_1, \dots, X_d$, second order polynomials in 
$X_1, \dots, X_d$ as well as feedforward neural networks with different activation functions. 

But irrespective of the method used to obtain an approximation of
$\bar{f}$, any Borel measurable candidate regression function $\hat{f} : \R^d \to \R$ 
yields an upper bound
\be \label{Unu}
U^{\nu} := \E^{\nu} \edg{\brak{Y - \hat{f}(X)}^2} 
\ee
of the minimal mean squared distance $D^{\nu}$. However, since
in typical applications, $U^{\nu}$ cannot be calculated exactly, we approximate it with a
Monte Carlo estimate 
\be \label{UnuN}
U^{\nu}_N := \frac{1}{N} \sum_{n =1}^{N} \brak{Y^n - \hat{f}(X^n)}^2
\ee
based on $N$ independent $\p^{\nu}$-realizations $(X^n, Y^n)_{n =1}^{N}$ of $(X,Y)$
drawn independently of any data $(\mathcal{X}^m, \mathcal{Y}^m)_{m = 1}^M$ used to 
determine $\hat{f}$.

Provided that $\E^{\nu} [ (Y -\hat{f}(X))^2] < \infty$, one obtains from the strong 
law of large numbers that 
\[
\lim_{N \to \infty} U^{\nu}_N = U^{\nu} \quad \p^{\nu}\mbox{-almost surely.}
\]
To derive confidence intervals, we compute the sample variance
\[
v^{U,\nu}_{N} := \frac{1}{N-1} \sum_{n=1}^{N} 
\brak{\brak{Y^n - \hat{f}(X^n)}^2 - U^{\nu}_N }^2
\]
and denote, for $\alpha \in (0,1)$, by $q_{\alpha}$ the $\alpha$-quantile of the standard normal 
distribution. Then the following holds.

\begin{lemma} \label{lemma:ciU}
Assume $\E^{\nu} \, Y^4 < \infty$ and $\E^{\nu} \, |\hat{f}(X)|^4 < \infty$. Then, 
for every $\alpha \in (1/2,1)$,
\be \label{conub}
\liminf_{N \to \infty} \p^{\nu} \edg{\abs{U^{\nu} - U^{\nu}_N} \le 
 q_{1 - \alpha}\sqrt{\frac{v^{U,\nu}_{N}}{N}} \; } 
\ge 1 - 2 \alpha.
\ee
\end{lemma}

\begin{proof}
In the special case where $Y = \hat{f}(X)$ $\p^{\nu}$-almost surely, one has 
$U^{\nu} = U^{\nu}_N = v^{U,\nu}_N = 0$ \, $\p^{\nu}$-almost surely for all $N \ge 1$. So
$\eqref{conub}$ holds trivially. On the other hand, if $\p^{\nu}[Y \neq \hat{f}(X)] > 0$, it
follows from the assumptions and the strong law of large numbers that $v^{U,\nu}_N$ 
converges $\p^{\nu}$-almost surely to $\Var^{\p^{\nu}} \brak{ (Y - \hat{f}(X))^2} > 0$ for $N \to \infty.$ Therefore, 
one obtains from the central limit theorem and Slutky's theorem that 
\[
\lim_{N \to \infty} \p^{\nu} \edg{\sqrt{\frac{N}{v^{U,\nu}_N}} \, 
\abs{U^{\nu} - U^{\nu}_{N}} \le q_{1-\alpha}} = 1- 2\alpha,
\] 
which shows \eqref{conub}.
\end{proof}

\section{Error estimates}
\label{sec:est}

Now, our goal is to derive bounds on the approximation error $\hat{f} - \bar{f}$ for a
given candidate regression function $\hat{f} \colon \R^d \to \R$. To do that
we assume in this section that $Y$ has a representation of the form:
\[
{\bf (R)} \qquad 
\begin{aligned}
& Y = h(X,V) \; \mbox{\sl  for a Borel function $h \colon \R^{d+k} \to \R$ and a 
$k$-dimensional}\\[-0.8mm]
& \mbox{\sl random vector $V$ that is independent of $X$ under $\p^{\nu}$.}
\end{aligned}
\]

\begin{Remark}
Provided that the probability space $(\Omega, {\cal F}, \p^{\nu})$ is rich enough,
$Y$ can always be assumed to be of the form (R). Indeed, if $(\Omega, {\cal F}, \p^{\nu})$
supports a random variable $V$ which, under $\p^{\nu}$, is uniformly distributed on 
the unit interval $(0,1)$ and independent of $X$, the function 
$h \colon \R^d \times (0,1) \to \R$ can be chosen as 
a conditional $\p^{\nu}$-quantile function of $Y$ given $X$ and extended to the rest of $\R^{d+1}$ arbitrarily. 
Then $(X,h(X,V))$ has the same $\p^{\nu}$-distribution as $(X,Y)$, and, in particular,
\[
\E^{\nu} [h(X,V) \mid X] = \bar{f}(X) \quad \mbox{$\nu$-almost surely.}
\]
However, for our method to be applicable, the function $h$ needs to be known explicitly.
\end{Remark}

A representation of the form (R) with a known function $h$ is
available in computational problems involving numerical regressions, such as 
regression methods to solve PDEs and BSDEs 
\citep[see, e.g.,][]{bally1997approximation, chevance1997numerical, BT, gobet2005regression, GobetT, FTW, beck2019deep}, SPDEs \citep[see, e.g.,][]{beck2020deep},
financial valuation problems \citep[see, e.g.][]{carriere96,longstaff2001valuing, tsitsiklis2001regression,broadie2004stochastic, broadie2008,becker2020pricing}
or financial risk management problems \citep[see, e.g.][]{lee2003computing, gordy2010nested, 
broadie2011efficient, BauerReussSinger, Cher}.

\subsection{Alternative representation of the minimal mean squared distance}
\label{sec:altrep}

The key ingredient of our approach is an alternative representation of the
minimal mean squared distance 
\be \label{Dnu} D^{\nu} = \min_{f \in {\cal B}(\R^d ; \, \R)} 
\E^{\nu} \edg{\brak{Y - f(X)}^2} = \E^{\nu} \edg{\brak{Y - \bar{f}(X)}^2}
\ee which does not involve 
a minimization problem or require knowledge of the true regression function $\bar{f}$
and, at the same time, can be approximated efficiently. An analogous representation 
exists for the squared $L^2$-norm of the conditional expectation 
\be \label{Cnu}
C^{\nu} := \E^{\nu} \, \bar{f}^2(X) = \no{ [\E^{\nu}[Y \mid X] }^2_{L^2(\p^{\nu})},
\ee
which will be helpful in the computation of relative approximation errors in 
Section \ref{sec:ex} below. If necessary, by enlarging\footnote{If assumption (R) holds,
e.g. the product space $(\Omega \times \Omega, {\cal F} \otimes {\cal F}, \p^{\nu} \otimes \p^{\nu})$
supports, next to $X$ and $V$, an independent copy $\tilde{V}$ of $V$.} the probability space
$(\Omega, {\cal F}, \p^{\nu}$), we can assume it supports a 
$k$-dimensional random vector $\tilde{V}$ that has the same 
$\p^{\nu}$-distribution as $V$ and is independent of $(X,V)$ under $\p^{\nu}$.
Let us define
 \[
Z := h(X,\tilde{V}).
\] 
Then, we have the following.

\begin{proposition} \label{prop:rep}
If $\E^{\nu} \, Y^2 < \infty$, then
\[
C^{\nu} = \E^{\nu} \edg{Y Z} \quad \mbox{and} \quad
D^{\nu} = \E^{\nu} \edg{Y (Y- Z)}.
\]
\end{proposition}

\begin{proof}
It follows from independence of $X$, $V$ and $\tilde{V}$ that
\[
\E^{\nu} \edg{Y Z} = \E^{\nu} \edg{\E^{\nu} \edg{h(X,V) h(X,\tilde{V}) \mid X}} =  
\E^{\nu} \edg{\bar{f}^2(X)} = C^{\nu}.
\]
Similarly, one has
\[
\E^{\nu} \edg{Y \bar{f}(X)} = \E^{\nu} \edg{\E^{\nu}[Y \mid X] \bar{f}(X) } = \E^{\nu} \edg{\bar{f}^2(X)},
\] 
from which one obtains
\[
\E^{\nu} \edg{Y (Y-Z)} = \E^{\nu} \edg{Y^2 - \bar{f}^2(X)} 
= \E^{\nu} \edg{Y^2 - 2 Y \bar{f}(X) + \bar{f}^2(X)} 
= \E^{\nu} \edg{\brak{Y - \bar{f}(X)}^2} = D^{\nu}.
\]
\end{proof}

\subsection{Approximation of $C^{\nu}$ and $D^{\nu}$}
\label{sec:apprCD}

To approximate $C^{\nu}$ and $D^{\nu}$, we use $\p^{\nu}$-realizations
$Z^n := h(X^n, \tilde{V}^n)$, $n = 1, \dots, N$, of $Z$
based on independent copies $\tilde{V}^n$ of $V$ drawn independently of
$(\mathcal{X}^m, \mathcal{Y}^m)$, $m = 1, \dots, M$, and $(X^n, Y^n, V^n)$, $n = 1, \dots, N$.
The corresponding Monte Carlo approximations of $C^{\nu}$ and $D^{\nu}$ are
\be \label{CDnuN}
C^{\nu}_N := \frac{1}{N} \sum_{n=1}^{N} Y^n Z^n
\quad \mbox{and} \quad
D^{\nu}_N := \frac{1}{N} \sum_{n=1}^{N} Y^n (Y^n- Z^n),
\ee
respectively. If $\E^{\nu} \, Y^2 < \infty$, then $\E^{\nu} \, Z^2 < \infty$ too, and 
one obtains from the strong law of large numbers that \[
\lim_{N \to \infty} C^{\nu}_N = C^{\nu} \quad \mbox{and} \quad
\lim_{N \to \infty} D^{\nu}_N = D^{\nu} \quad \p^{\nu}\mbox{-almost surely.}
\] Moreover,
for the sample variances
\[ 
v^{C,\nu}_N := \frac{1}{N-1} \sum_{n=1}^{N}
\brak{Y^n Z^n - C^{\nu}_N}^2 \quad \mbox{and} \quad
v^{D,\nu}_N := \frac{1}{N-1} \sum_{n=1}^{N}
\brak{Y^n (Y^n-Z^n) - D^{\nu}_N}^2,
\]
the following analog of Lemma \ref{lemma:ciU} holds.

\begin{lemma} \label{lemma:ciD}
If $\E^{\nu} \, Y^4 < \infty$, then, for every $\alpha \in (1/2,1)$, 
\be \label{conc}
\liminf_{N \to \infty} \p^{\nu} \edg{\abs{C^{\nu} - C^{\nu}_N} \le 
 q_{1 - \alpha}\sqrt{\frac{v^{C,\nu}_{N}}{N}} \; } 
\ge 1 - 2 \alpha \ee
and
\be \label{cond}
\liminf_{N \to \infty} \p^{\nu} \edg{\abs{D^{\nu} - D^{\nu}_N} \le 
 q_{1 - \alpha}\sqrt{\frac{v^{D,\nu}_{N}}{N}} \; } 
\ge 1 - 2 \alpha.
\ee
\end{lemma}

\begin{proof}
If $C^{\nu} = Y Z$ $\p^{\nu}$-almost surely, then 
$C^{\nu} - C^{\nu}_N = v^{C,\nu}_N = 0$ $\p^{\nu}$-almost surely for all $N \ge 1$, and
$\eqref{conc}$ is immediate. On the other hand, if $\p^{\nu}[C^{\nu} \neq YZ] > 0$, 
one obtains from the strong law of large numbers that 
$v^{C,\nu}_N \to \Var^{\p{\nu}} \brak{YZ} > 0$
$\p^{\nu}$-almost surely for $N \to \infty$, and it follows 
from the central limit theorem together with Slutky's theorem that 
\[
\lim_{N \to \infty} \p^{\nu} \edg{\sqrt{\frac{N}{v^{C,\nu}_N}} \, 
\abs{C^{\nu} - C^{\nu}_{N}} \le q_{1-\alpha}} = 1- 2\alpha.
\] 
This shows \eqref{conc}. \eqref{cond} follows analogously.
\end{proof}

\subsection{$L^2$-bounds on the approximation error}
\label{sec:L2est}

We now derive $L^2$-bounds on the error resulting from approximating the 
true regression function $\bar{f}$ with a candidate regression function $\hat{f}$.
Let us denote by $L^2(\nu)$ the space of all Borel functions $f \colon \R^d \to \R$ 
satisfying 
\[
\n{f}^2_{L^2(\nu)} := \E^{\nu} f^2(X) = \int_{\R^d} f^2(x) d \nu(x) < \infty
\]
and consider the squared $L^2(\nu)$-norm of the approximation error
\be \label{Fnu}
F^{\nu} := \n{\hat{f} - \bar{f}}^2_{L^2(\nu)} = \E^{\nu} \edg{\brak{\hat{f}(X) - \bar{f}(X)}^2}.
\ee
$F^{\nu}$ has the following alternative representation.

\begin{theorem} \label{thm:error}
If $\E^{\nu} \, Y^2 < \infty$ and $\E^{\nu} \hat{f}^2(X) < \infty$, then
\be \label{F}
F^{\nu} = \E^{\nu} \edg{YZ + \hat{f}(X) \brak{\hat{f}(X) - Y - Z}}.
\ee
\end{theorem}

\begin{proof}
Since $\bar{f}(X) = \E^{\nu}[Y \mid X]$, it follows from $\E^{\nu} \, Y^2 < \infty$ 
and the conditional Jensen inequality that $\E^{\nu} \bar{f}^2(X) < \infty$ as well. Furthermore,
$Y - \bar{f}(X)$ is orthogonal to $\hat{f}(X) - \bar{f}(X)$ in $L^2(\p^{\nu})$. 
Therefore, one obtains from Pythagoras' theorem that
\be \label{Pyth}
F^{\nu} = \no{\hat{f}(X) - \bar{f}(X)}^2_{L^2(\p^{\nu})}
= \no{Y - \hat{f}(X)}^2_{L^2(\p^{\nu})} - \no{Y - \bar{f}(X)}^2_{L^2(\p^{\nu})}.
\ee
In addition, we know from Proposition \ref{prop:rep} that
\be \label{prop}
\no{Y - \bar{f}(X)}^2_{L^2(\p^{\nu})} = \E^{\nu} \edg{\brak{Y - \bar{f}(X)}^2} = \E^{\nu}[Y(Y-Z)].
\ee
So, since
\[
\E^{\nu} \edg{Y \hat{f}(X)} = \E^{\nu} \edg{\E^{\nu}[h(X,V) \mid X] \, \hat{f}(X)} 
= \E^{\nu} \edg{\E^{\nu}[h(X,\tilde{V}) \mid X] \, \hat{f}(X)} = \E^{\nu} \edg{Z \hat{f}(X)},
\]
we obtain from \eqref{Pyth} und \eqref{prop} that 
\[
F^{\nu} = \E^{\nu} \edg{Y^2 - (Y+Z) \hat{f}(X) + \hat{f}^2(X)- Y(Y-Z)} = 
\E^{\nu} \edg{YZ + \hat{f}(X) \brak{\hat{f}(X) - Y - Z}},
\]
which shows \eqref{F}.
\end{proof}

In view of \eqref{F}, we approximate $F^{\nu}$ with the Monte Carlo average
\be
\label{FnuN}
F^{\nu}_N := \frac{1}{N} \sum_{n = 1}^{N} 
\crl{Y^n Z^n + \hat{f}(X^n) \brak{\hat{f}(X^n) - Y^n - Z^n}}
\ee
and denote the corresponding sample variance by
\[
v^{F, \nu}_N := \frac{1}{N-1} \sum_{n=1}^{N} 
\crl{Y^n Z^n + \hat{f}(X^n) \brak{\hat{f}(X^n) - Y^n - Z^n} - F^{\nu}_N}^2.
\]
The following lemma provides approximate confidence upper bounds for the true 
squared $L^2$-approximation error \eqref{Fnu}.

\begin{lemma} \label{lemma:Fci}
If $\E^{\nu} \, Y^4 < \infty$ and $\E^{\nu} \hat{f}^4(X) < \infty$, one has for all $\alpha \in (0,1)$,
\be \label{ciF}
\liminf_{N \to \infty}
\p^{\nu} \edg{F^{\nu} \le F^{\nu}_N + q_{\alpha} \, \sqrt{\frac{v^{F,\nu}_N}{N}} \;}
\ge \alpha.
\ee
\end{lemma}

\begin{proof}
In the special case, where
\[
YZ + \hat{f}(X)(\hat{f}(X) - Y - Z) = F^{\nu} \quad \mbox{$\p^{\nu}$-almost surely,}
\]
one has
\[
F^{\nu} - F^{\nu}_N = v^{F,\nu}_N = 0 \quad 
\mbox{$\p^{\nu}$-almost surely for all } N \ge 1,
\]
and \eqref{ciF} is clear. On the other hand, if 
\[
\p^{\nu} \edg{YZ + \hat{f}(X)(\hat{f}(X) - Y - Z) \neq F^{\nu}} > 0,
\]
it follows from the strong law of large numbers that
\[
\lim_{N \to \infty} v^{F,\nu}_N = {\rm Var}^{\p^{\nu}} 
\brak{YZ + \hat{f}(X)(\hat{f}(X) - Y - Z} > 0 \quad \mbox{$\p^{\nu}$-almost surely.}
\]
So, one obtains from the central limit theorem and Slutky's theorem that
\[
\liminf_{N \to \infty} \p^{\nu} \edg{\sqrt{\frac{N}{v^{F,\nu}_N}} \brak{F^{\nu} - F^{\nu}_N} \le
q_{\alpha}} = \alpha,
\]
which implies \eqref{ciF}.
\end{proof}

In applications where $\bar{f}$ needs to be approximated well at a 
given point $x_0$ in the support of the distribution $\nu_X$ of $X$, 
$\nu_X$ can be distorted so as to obtain a probability measure 
$\nu \ll \nu_X$ on $\R^d$ that concentrates more weight around $x_0$. Then, 
provided that $\hat{f} - \bar{f}$ is continuous at $x_0$, $\n{\hat{f} - \bar{f}}_{L^2(\nu)}$ 
approximates the point-wise difference $|\hat{f}(x_0) - \bar{f}(x_0)|$. More precisely, if 
$\n{.}_2$ denotes the standard Euclidean norm on $\R^d$, the following holds.

\begin{lemma} \label{lemma:point}
Assume $\E \, Y^2 < \infty$, $\E \, \hat{f}^2(X) < \infty$ and $\hat{f} - \bar{f}$ 
is continuous at a point $x_0 \in \R^d$. Let
$(\nu_n)_{n \ge 1}$ be a sequence of Borel
probability measures on $\R^d$ given by $d \nu_n/d\nu_X = p_n$ for a sequence 
of Borel functions $p_n \colon \R^d \to [0,\infty)$ satisfying
\[
\int_{\R^d} p_n(x) d\nu_X(x) = 1 \; \mbox{ for all } n \ge 1 \quad \mbox{and} \quad
\sup_{x \in \R^d \, : \, \n{x - x_0}_2 > 1/n} p_n(x) \to 0 \; \mbox{ for $n \to \infty$}.
\]
Then 
\[
\lim_{n \to \infty} \n{\hat{f} - \bar{f}}_{L^2(\nu_n)} = |\hat{f}(x_0) - \bar{f}(x_0)|.
\]
\end{lemma}

\begin{proof} 
It follows from $\E \, Y^2 < \infty$ that $\E \, \bar{f}^2(X) < \infty$, which together 
with the condition $\E \, \hat{f}^2(X) < \infty$, implies that $f := \hat{f} - \bar{f} \in L^2(\nu_X)$.
Moreover, one obtains from the assumptions that for every $\varepsilon > 0$, 
there exists an $n \ge 1$ such that
\[
\abs{f(x) - f(x_0)} \le \varepsilon \quad \mbox{for all }
x \in \R^d \mbox{ satisfying } \n{x - x_0}_2 \le 1/n,
\]
and
\[
\int_{\crl{x \in \R^d \, : \, \n{x - x_0}_2 > 1/n}} \brak{f(x) - f(x_0)}^2 d\nu_n(x) 
= \int_{\crl{x \in \R^d \, : \, \n{x - x_0}_2 > 1/n}} \brak{f(x) - f(x_0)}^2 p_n(x) d\nu_X(x) \le \varepsilon^2.
\]
Hence, 
\beas
\n{f - f(x_0)}^2_{L^2(\nu_n)} &=& \int_{\crl{x \in \R^d \, : \, \n{x - x_0}_2 \le 1/n}} \brak{f(x) - f(x_0)}^2 d\nu_n(x)\\
&&+ \int_{\crl{x \in \R^d \, : \, \n{x - x_0}_2 > 1/n}} \brak{f(x) - f(x_0)}^2 d\nu_n(x) \le 2 \varepsilon^2,
\eeas
and therefore,
\[
\abs{\n{f}_{L^2(\nu_n)} - |f(x_0)|} \le \n{f - f(x_0)}_{L^2(\nu_n)} \le \sqrt{2}\, \varepsilon.
\]
Since $\varepsilon > 0$ was arbitrary, this proves the lemma.
\end{proof}

\section{Examples}
\label{sec:ex}

In all our examples we compute a candidate regression function $\hat{f} \colon \R^d \to \R$
by minimizing an empirical mean squared distance of the form \eqref{empD}. 
For comparison reasons we minimize \eqref{empD} over different families 
${\cal S}$ of Borel measurable functions $f \colon \R^d \to \R$, in each case using a numerical method suited to the specific form of ${\cal S}$.

\begin{enumerate}
\item 
First, we use linear regression on $1, X_1, \dots, X_d$. The corresponding 
function family ${\cal S}$ consists of all linear combinations of $1, x_1, \dots, x_d$, 
and the minimization of the empirical mean squared distance \eqref{empD} becomes 
the ordinary least squares problem 
\[
\min_{\beta \in \R^{d+1}} \sum_{m=1}^M \brak{\mathcal{Y}^m - \beta_0 - \sum_{i=1}^d
\beta_i \mathcal{X}^m_i}^2.
\]
This yields a candidate regression function of the form
$\hat{f}(x) = \hat{\beta}_0 + \sum_{i=1}^d \hat{\beta}_i x_i$, where 
$\hat{\beta} \in \R^{d+1}$ is a solution of the normal equation
\be \label{neq}
A^T A \, \hat{\beta} = A^T y
\ee
for $A \in \R^{M \times (d+1)}$ and $y \in \R^M$ given by
\[
A = 
\left(
\begin{array}{cccc}
1 & \mathcal{X}^1_1 & ... & \mathcal{X}^1_d\\
1 & \mathcal{X}^2_1 & ... & \mathcal{X}^2_d\\
... & ... & ... & ...\\
1 & \mathcal{X}^M_1 & ... & \mathcal{X}^M_d
\end{array}
\right) \quad \mbox{and} \quad y = \left(
\begin{array}{cccc}
\mathcal{Y}^1\\
\dots \\
\dots \\
\mathcal{Y}^M
\end{array}
\right).
\]
In Sections \ref{ex:poly} and \ref{ex:non-poly} we use $M = 2 \times 10^6$ 
independent Monte Carlo simulations $(\mathcal{X}^m, \mathcal{Y}^m)$ for the linear regression,
while in Sections \ref{ex:max-call} and \ref{ex:binary}, where the 
examples are higher-dimensional, we use $M = 5 \times 10^5$ of them.
If the matrix $A^T A$ is invertible, equation \eqref{neq} has a unique 
solution given by $\hat{\beta} = (A^T A)^{-1} A^T y$. If $A^T A$ is invertible and,
in addition, well-conditioned, $\hat{\beta}$ can efficiently be computed using the 
Cholesky decomposition $R^TR$ of $A^T A$
to solve $R^T z = A^T y$ and $R \hat{\beta} = z$ in two steps. On the other hand, 
if $A^T A$ is not invertible or ill-conditioned, the Cholesky method is numerically unstable.
In this case, we compute a singular value decomposition $U \Sigma V^T$ of $A$ for orthogonal 
matrices $U \in \R^{M \times M}$,
$V \in \R^{(d+1) \times (d+1)}$ and a diagonal matrix $\Lambda \in \R^{M \times (d+1)}$
with diagonal entries $\lambda_1 \ge \lambda_2 \ge \dots \ge \lambda_{d+1} \ge 0$.
The solution of \eqref{neq} with the smallest Euclidean norm is then given by 
\[
\hat{\beta} = V \left( \begin{array}{ccccc}
\lambda_1^{-1} 1_{\crl{\lambda_1 > 0}} & 0 & \dots & \dots & 0\\
0 & \lambda^{-1}_2 1_{\crl{\lambda_2 > 0}} & \dots & \dots & 0\\
0 & \dots & \dots & \dots & 0\\
0 & \dots & \dots & \dots & \lambda_{d+1}^{-1} 1_{\crl{\lambda_{d+1} > 0}}\\
0 & \dots & \dots & \dots & 0\\
\dots & \dots & \dots & \dots & \dots\\
0 & \dots & \dots & \dots & 0
\end{array} 
\right) U^T y,
\]
which, for numerical stability reasons, we approximate with a
truncated SVD solution
\be \label{pinv}
\hat{\beta}_c = V \left( \begin{array}{ccccc}
\lambda_1^{-1} 1_{\crl{\lambda_1 > c}} & 0 & \dots & \dots & 0\\
0 & \lambda^{-1}_2 1_{\crl{\lambda_2 > c}} & \dots & \dots & 0\\
0 & \dots & \dots & \dots & 0\\
0 & \dots & \dots & \dots & \lambda_{d+1}^{-1} 1_{\crl{\lambda_{d+1} > c}}\\
0 & \dots & \dots & \dots & 0\\
\dots & \dots & \dots & \dots & \dots\\
0 & \dots & \dots & \dots & 0
\end{array} 
\right) U^T y
\ee
for a small cutoff value $c > 0$; see, e.g., \citet{bjorck96}.

\item
As second method we use second order polynomial regression; that is, we
regress on $1, X_1, \dots, X_d$ and all second order terms $X_i X_j$, $1 \le i \le j \le d$. 
${\cal S}$ is then the linear span of $1, x_1, \dots, x_d$ and $x_ix_j$, 
$1 \le i \le j \le d$, and a canditate regression function 
can be computed as in 1. above, except that now the feature matrix $A$ has
$1 + 3d/2 + d^2/2$ columns. As before, we use $M = 2 \times 10^6$ 
independent Monte Carlo simulations $(\mathcal{X}^m, \mathcal{Y}^m)$ in Sections \ref{ex:poly} -- \ref{ex:non-poly} and $M = 5 \times 10^5$ of them in Sections 
\ref{ex:max-call} -- \ref{ex:binary}, and again, 
we use the Cholesky decomposition of $A^T A$ 
to solve \eqref{neq} if $A^T A$ is well-conditioned and a truncated 
SVD solution otherwise\footnote{We used Cholesky decomposition 
for the linear and polynomial regressions without additional feature
in Sections \ref{ex:poly}--\ref{ex:non-poly} and the pseudoinversion \eqref{pinv} 
based on truncated SVD for all other linear and polynomial regressions in Section \ref{sec:ex}.
We computed \eqref{pinv} with a standard pseudoinverse 
command. In most examples the default cutoff value $c$ gave good results. 
In the high-dimensional examples of Sections \ref{ex:max-call}--\ref{ex:binary}
a slightly higher cutoff value $c$ improved the results of the polynomial regressions. 
Alternatively, one could use ridge regression with a suitable penalty parameter
or (stochastic) gradient descent to solve the least squares problem
in cases where $A^T A$ is ill-conditioned.
}.

\item 
In our third method, ${\cal S}$ consists of all neural networks of a given architecture. 
In this paper we focus on feedforward neural networks of the form
\be \label{nn}
f^{\theta} = 
a_{D}^{\theta}\circ \varphi \circ a_{D-1}^{\theta} \circ \dots \circ \varphi \circ a_{1}^{\theta},
\ee
where 
\begin{itemize} 
\item $D \ge 1$ is the depth of the network;
\item $a^{\theta}_{i} : \mathbb{R}^{q_{i-1}} \to \mathbb{R}^{q_{i}}$, $i = 1, \dots, D$, 
are affine transformations of the form $a^{\theta}_{i}(x) = A_ix + b_i$ for matrices 
$A_i \in \mathbb{R}^{q_i \times q_{i-1}}$ 
and vectors $b_i \in \mathbb{R}^{q_i}$, where $q_0$ is the input dimension, $q_D = 1$ 
the output dimension and $q_i$, $i = 1, \dots, D - 1$, the number of neurons in the $i$-th hidden layer;\item 
$\varphi \colon \R \to \R$ is a non-linear activation function applied component-wise in each hidden layer.
\end{itemize}
In the examples below, we use networks of depth $D = 4$ and 128 neurons in each of the 
three hidden layers. We compare the commonly used activation functions $\varphi = \tanh$
and ${\rm ReLU}(x) := \max\{0,x\}$ to the following smooth version of 
${\rm Leaky ReLU}(x) := \max \{\alpha x, x\}$:
\[
{\rm LSE}(x) := \log \brak{e^{\alpha x} + e^x} \quad \mbox{for } \alpha = 0.01,
\]
which is efficient to evaluate numerically and, by the LogSumExp inequality, satisfies
\[
{\rm LeakyReLU}(x) \leq {\rm LSE}(x) \leq {\rm LeakyReLU}(x) + \log{2}.
\]
In addition, ${\rm LSE}$ is everywhere differentiable with non-vanishing derivative, which alleviates the problem of 
vanishing gradients that can arise in the training of $\tanh$ and ${\rm ReLU}$
networks. We initialize the parameter vector $\theta$ according to Xavier initialization
\cite{glorot2010understanding} and then optimize it by iteratively decreasing the empirical 
mean squared distance \eqref{empD} with Adam stochastic gradient descent \cite{kingma2014adam} using mini-batches of size $2^{13}$ and 
batch-normalization\footnote{Note that while the trained network is 
of the form \eqref{nn}, training with batch-normalization decomposes each affine 
transformation into a concatenation $a^{\theta}_i = a^{\theta}_{i,2} \circ a^{\theta}_{i,1}$
for a general affine transformation $a^{\theta}_{i,1} \colon \R^{q_{i-1}} \to \R^{q_i}$
and a batch-normalization transformation $a^{\theta}_{i,2} \colon \R^{q_i} \to \R^{q_i}$, 
both of which are learned from the data. This usually stabilizes the training process but increases 
the number of parameters that need to be learned.}
\cite{IoffeSzegedy} before each activation $\varphi$.
We perform 250,000 gradient steps with standard Adam parameters, except that we start with  
a learning rate of $0.1$, which we manually reduce to $0.05$, $10^{-2}$,
$10^{-3}$, $10^{-4}$, $10^{-5}$ and $10^{-6}$ after 1000, 5000, 25,000, 50,000, 
100,000 and 150,000 iterations, respectively.
To avoid slow cross device communications between 
the CPU and GPU, we generate all 
simulations on the fly during the training procedure.
Since we simulate from a model, we can produce a large training set 
and therefore, do not need to worry about overfitting to the training data.
\end{enumerate}

\begin{Remark}
In many applications, the performance of the numerical regression can be improved with
little additional effort by adding a redundant feature of the form $a(X)$ for a Borel 
measurable function $a \colon \R^d \to \R$ capturing important aspects of the relation between
$X$ and $Y$. For instance, if $Y$ is given by $Y = h(X,V)$ for a Borel function $h \colon \R^{d+k} \to \R$ and a 
$k$-dimensional random vector $V$, adding the additional feature $a(X) = h(X,0)$, or something 
similar, often yields good results. Instead of minimizing the mean squared distance \eqref{D}, 
one then tries to find a Borel function $\hat{g} \colon \R^{d+1} \to \R$ that minimizes
$\E^{\nu} \edg{\brak{Y - \hat{g}(X, a(X))}^2}$ and approximates the regression function
$\bar{f} \colon \R^d \to \R$ with $\hat{f}(x) = \hat{g}(x, a(x))$, $x \in \R^d$.
\end{Remark}

In all examples, we report for all different methods used to determine a candidate regression function $\hat{f}$,
\begin{itemize}
\item an approximate 95\% confidence interval for $U^{\nu} = \E^{\nu} \edg{\brak{Y - \hat{f}(X)}^2}$ using \eqref{conub}.
\item an approximate 95\% confidence interval for $D^{\nu} = \E^{\nu} \edg{\brak{Y - \bar{f}(X)}^2}$
using \eqref{cond}.
\item an estimate of the relative error $\n{\hat{f} - \bar{f}}_{L^2(\nu)}/ \n{\bar{f}}_{L^2(\nu)}$
of the form $\sqrt{F^{\nu}_N/C^{\nu}_N}$ for $C^{\nu}_N$ and $F^{\nu}_N$ given in 
\eqref{CDnuN} and \eqref{FnuN}, respectively. Note that while the theoretical values 
$C^{\nu} = \n{\bar{f}}_{L^2(\nu)}$ and $F^{\nu} = \n{\hat{f} - \bar{f}}^2_{L^2(\nu)}$ are both
non-negative, in some of our examples, $F^{\nu}$ is close to zero. So due to Monte Carlo noise,
the estimate $F^{\nu}_N$ can become negative. In these cases, we report 
$- \sqrt{- F^{\nu}_N/C^{\nu}_N}$ instead of $\sqrt{F^{\nu}_N/C^{\nu}_N}$.
\item 
an approximate 95\% confidence upper bound for the error $\n{\hat{f} - \bar{f}}_{L^2(\nu)}$ based on 
\eqref{ciF} expressed as a fraction of $\n{\bar{f}}_{L^2(\nu)}$ as estimated by $C^{\nu}_N$ given in \eqref{CDnuN}.
\item 
The time in seconds it took to compute the approximate regression function $\hat{f}$.
\end{itemize}

In Sections \ref{ex:poly} and \ref{ex:non-poly} below we used
$N = 6 \times 10^8$ independent Monte Carlo simulations $(X^n, Y^n, Z^n)$ 
to compute the estimates 
$U^{\nu}_N$, $D^{\nu}_N$, $F^{\nu}_N$, $C^{\nu}_N$ together with the corresponding
confidence intervals, whereas in Sections \ref{ex:max-call} and \ref{ex:binary}, 
due to the higher dimensionality of the examples, we only worked with $N = 6 \times 10^7$ 
independent Monte Carlo simulations. 
To fit such large test data sets into the computer memory, we split them into 
6,000 independent batches of 100,000 or 10,000 data points, respectively. 
In most examples, we chose $\nu$ to be equal to the
original distribution $\nu_X$ of $X$, in which case $\E^{\nu}$ equals $\E$. 

All computations were performed on a Nvidia 
GeForce RTX 2080 Ti GPU together with Intel Core Xeon CPUs using 
\classname{Python 3.9.6}, \classname{TensorFlow 2.5.0} with eager mode disabled
and \classname{TensorFlow Probability 0.13.0} on \classname{Fedora 32}. 

\subsection{A four-dimensional polynomial example}
\label{ex:poly}

In our first example, we consider a simple non-linear model for $(X,Y)$ in which the conditional 
expectation $\E[Y \mid X]$ can be computed explicitly. This enables us to benchmark our numerical 
results against the theoretical values. Let $X = (X_1,X_2,X_3,X_4)$ 
be a four-dimensional random vector and $V$, $Y$ random variables such that 
$X_1, \dots, X_4, V$ are i.i.d.\ standard normal and $Y$ is of the form
\be \label{Ypoly}
Y = X_1 + X_2^2 + X_3 X_4 + V.
\ee
Then the conditional expectation is
\be \label{trueregr}
\E[Y \mid X]=X_1+ X_2^2 + X_3 X_4,
\ee
from which the minimal mean squared distance under $\p$ can be seen to be
\[
D^{\nu_X} = \E \edg{\brak{Y - \E[Y \mid X]}^2}  = \E \edg{V^2} = 1.\] 
Replacing $V$ by $0$ in the expression \eqref{Ypoly} would suggest to use the
additional feature $a(X) = X_1 + X^2_2 + X_3 X_4$. However, since 
this would directly solve the problem, we are not using it in this example.

Our numerical results are listed in Table \ref{tab:1}. 
More details are provided in Table \ref{tab:7} in the Appendix.
As could be expected, since the true regression function 
\eqref{trueregr} is a second order polynomial, the accuracy of the linear regression is poor, 
while the the second order polynomial regression 
works very well. All three neural networks provide results comparable to the one of 
the second order polynomial regression, albeit with more computational effort.

\begin{table}[h!]
\centering
{\footnotesize
\begin{tabular}{l|c|c|c|c|c}
\thead{}  &  95\% CI  $U^{\nu_X}$ &  95\% CI $D^{\nu_X}$ & 
\(\displaystyle \frac{\n{\hat{f} - \bar{f}}_{L^2(\nu_X)}}{\n{\bar{f}}_{L^2(\nu_X)}} \)
& \( \displaystyle \frac{\mbox{\footnotesize 95\% CB } \n{\hat{f} - \bar{f}}_{L^2(\nu_X)}}{\n{\bar{f}}_{L^2(\nu_X)}} \) & \begin{minipage}{9mm} comp. time\\ for $\hat{f}$ \end{minipage} \cr
\midrule
\makecell{lin. regr.}   &    [3.99957, 4.00105] &  [0.99982, 1.00040] &  77.46 \% & 77.47 \% & 0.1 s \cr
\myrowcolour
\makecell{poly. regr.}   
  &    [0.99979, 1.00002] &  [0.99982, 1.00040] & 0.25 \% &              0.68 \% &  0.1 s \cr
\makecell{NN tanh}    &   [0.99986, 1.00009] &  [0.99982, 1.00040] & 0.45 \% & 0.77 \%  & 1332.0 s  \cr
\myrowcolour
\makecell{NN ReLU}    &   [1.00007, 1.00030] &  [0.99982, 1.00040] &  0.79 \% & 1.01 \% & 1328.3 s  \cr
\makecell{NN LSE}   &   [0.99988, 1.00010] &  [0.99982, 1.00040] & 0.47 \% &              0.79 \% & 1483.2 s  \cr
\bottomrule
\end{tabular}
}
\caption{Numerical results for the polynomial example \eqref{Ypoly}}
\label{tab:1}
\end{table}

\subsection{A five-dimensional non-polynomial example}
\label{ex:non-poly}

In our second example, we consider a non-polynomial relationship between 
$Y$ and $X$. More precisely, we let $X_1, V_1, \dots, X_5, V_5$ be i.i.d.\ standard 
normal and assume that $Y$
is of the form
\be \label{non-poly}
Y = 5 \log \brak{5 + (X_1+V_1)^2 X_2^2 + V^2_2} \tanh \brak{(X_3 + V_3) (X_4 + V_4) (X_5 + V_5)^2}.
\ee
Then the conditional expectation $\E[Y \mid X]$ is not known in closed form. 
Setting $V_1 = \dots = V_5 = 0$ in 
\eqref{non-poly} suggests to use the additional feature 
\be \label{adnon-poly}
a(X) = 5 \log \brak{5 + X_1^2 X_2^2} \tanh \brak{X_3 X_4 X_5^2}.
\ee

\subsubsection{Minimizing the mean squared distance under $\p$}
\label{ex:non-polyP}

We first search for the function $f \colon \R^d \to \R$ minimizing the mean
squared distance $\E [\brak{Y - f(X)}^2]$ under the original measure $\p$.
The numerical results are reported in Table \ref{tab:2}, and
more details can be found in Table \ref{tab:8} in the Appendix.
It can be seen that the second order polynomial regression yields better results
than the linear regression, but now, both are clearly outperformed by the
three neural network approaches. Moreover, the inclusion of the 
additional feature \eqref{adnon-poly} improves the accuracy of the
linear and second order polynomial regressions, while it does not increase the 
performance of the neural networks significantly.

\begin{table}[h!]
\centering
{\footnotesize
\begin{tabular}{l|c|c|c|c|c}
\thead{}  &  95\% CI  $U^{\nu_X}$ &  95\% CI $D^{\nu_X}$ & 
\(\displaystyle \frac{\n{\hat{f} - \bar{f}}_{L^2(\nu_X)}}{\n{\bar{f}}_{L^2(\nu_X)}} \)
& \( \displaystyle \frac{\mbox{\footnotesize 95\% CB } \n{\hat{f} - \bar{f}}_{L^2(\nu_X)}}{\n{\bar{f}}_{L^2(\nu_X)}} \) & \begin{minipage}{9mm} comp. time\\ for $\hat{f}$ \end{minipage} \cr
\midrule
\makecell{lin. regr.}   &  [41.57390, 41.58212] &  [36.16566, 36.17592] &  100.00 \% & 100.02 \%  &        0.1 s \cr
\myrowcolour
\makecell{lin. regr., add. feature}    
& [37.89192, 37.90030] &  [36.16566, 36.17592] &                56.48 \% &             56.52 \%   &      0.1 s   \cr
\makecell{poly. regr.}     &  [36.66939, 36.67730] &  [36.16566, 36.17592] &                30.47 \% &             30.55 \%  &          0.2 s\cr
\myrowcolour
\makecell{poly. regr., add. feature}   &  [36.39638, 36.40441] &  [36.16566, 36.17592] &                20.58 \% &             20.70 \%   &       0.2 s \cr
\makecell{NN tanh}   &  [36.16812, 36.17617] &  [36.16566, 36.17592] &                 0.90 \% &              2.44 \%  &        1410.3 s \cr
\myrowcolour
\makecell{NN tanh, add. feature}   &  [36.16801, 36.17606] &  [36.16566, 36.17592] &  0.75 \% &              2.38 \%  &     1440.5 s   \cr
\makecell{NN ReLU}   &  [36.16800, 36.17605] &  [36.16566, 36.17592] & 0.74 \% &  2.38 \%   & 1399.0 s \cr
\myrowcolour
\makecell{NN ReLU, add. feature}   
&  [36.16775, 36.17579] &  [36.16566, 36.17592] &  0.15 \% &   2.27 \%     &       1470.0 s \cr
\makecell{NN LSE}   &  [36.16757, 36.17562] &  [36.16566, 36.17592] & -0.49 \% &    2.21 \%    &        1579.7 s \cr
\myrowcolour
\makecell{NN LSE, add. feature}    
&  [36.16764, 36.17569] &  [36.16566, 36.17592] & -0.36 \% &              2.24 \%   &     1532.1 s   \cr
\bottomrule
\end{tabular}
}
\caption{Numerical results for the non-polynomial example \eqref{non-poly} regressed 
under $\p$}
\label{tab:2}
\end{table}

\subsubsection{Minimizing the mean squared distance under a distorted measure $\p^{\nu}$}
\label{ex:non-polydistorted}

As a variant, we numerically minimize the mean squared distance 
$\E^{\nu} [\brak{Y - f(X)}^2]$ in the model \eqref{non-poly} with respect to a
distorted measure $\p^{\nu}$ under which $X_1, V_1 \dots, X_5, V_5$ are independent 
with $X_1, \dots X_5 \sim N(1,1/10)$ and $V_1, \dots, V_5 \sim N(0,1)$. The measure 
$\nu$ concentrates more mass around the point $(1,\dots, 1) \in \R^5$ than the original 
distribution $\nu_X$ of $X$. But since $(V_1, \dots, V_5)$ has the same distribution 
as under $\p$, the minimizing function $f$ coincides with the same theoretical 
regression function $\bar{f}$ as before. 
However, $L^2$-norms are now measured with respect to $\p^{\nu}$ instead of $\p$, which 
leads to different numerical results in Tables \ref{tab:3} and \ref{tab:9} compared to Tables \ref{tab:2} and \ref{tab:8}.
It can be seen that, as before in the $\p$-minimization, the three neural networks provide 
better results than the second order polynomial regression, which works better than the linear 
regression. But now, including the additional feature \eqref{adnon-poly} only improves the 
accuracy of the linear regression slightly, while it does not help the other methods.

\begin{table}[h!]
\centering
{\footnotesize
\begin{tabular}{l|c|c|c|c|c}
\thead{}  &  95\% CI  $U^{\nu}$ &  95\% CI $D^{\nu}$ & 
\(\displaystyle \frac{\n{\hat{f} - \bar{f}}_{L^2(\nu)}}{\n{\bar{f}}_{L^2(\nu)}} \)
& \( \displaystyle \frac{\mbox{\footnotesize 95\% CB } \n{\hat{f} - \bar{f}}_{L^2(\nu)}}{\n{\bar{f}}_{L^2(\nu)}} \) & \begin{minipage}{9mm} comp. time\\ for $\hat{f}$ \end{minipage} \cr
\midrule
\makecell{lin. regr.}   &  [39.93194, 39.94025] &  [39.84392, 39.85452] &              8.75 \% &              8.89 \%  &      0.1 s  \cr
\myrowcolour
\makecell{lin. regr., add. feature}   &  [39.92223, 39.93055] &  [39.84392, 39.85452] &                 8.24 \% &              8.39 \%  &      0.1 s \cr
\makecell{poly. regr.}     &  [39.85144, 39.85974] &  [39.84392, 39.85452] &                   2.01 \% &              2.56 \%     &      0.2 s  \cr
\myrowcolour
\makecell{poly. regr., add. feature}   &  [39.85101, 39.85930] &  [39.84392, 39.85452] & 1.92 \% &              2.48 \%   &     0.2 s \cr
\makecell{NN tanh}   &  [39.84711, 39.85541] &  [39.84392, 39.85452] &                 0.33 \% &              1.62 \%   &      1373 s    \cr
\myrowcolour
\makecell{NN tanh, add. feature}   &  [39.84716, 39.85546] &  [39.84392, 39.85452] &                  0.40 \% &              1.63 \%    &       1449 s   \cr
\makecell{NN ReLU}   &  [39.84717, 39.85547] &  [39.84392, 39.85452] &                 0.41 \% &              1.63 \%   &     1390 s      \cr
\myrowcolour
\makecell{NN ReLU, add. feature}   &  [39.84724, 39.85554] &  [39.84392, 39.85452] &   0.48 \% &              1.65 \%    &      1411 s    \cr
\makecell{NN LSE}   &   [39.84710, 39.85541] &  [39.84392, 39.85452] & 0.32 \% &               1.62 \%    &        1515 s  \cr
\myrowcolour
\makecell{NN LSE, add. feature}     &  [39.84717, 39.85547] &  [39.84392, 39.85452] &  0.40 \% &              1.63 \%   &     1561 s    \cr
\bottomrule
\end{tabular}
}
\caption{Numerical results for the non-polynomial example \eqref{non-poly} 
regressed under $\p^{\nu}$}
\label{tab:3}
\end{table}

\subsection{Max-call options}
\label{ex:max-call}

Different pricing and risk management problems require a conditional valuation of
a financial product conditional on the state of the world at a later time
\citep[see, e.g.][]{carriere96,longstaff2001valuing, tsitsiklis2001regression,broadie2004stochastic, broadie2008,becker2020pricing, lee2003computing, gordy2010nested, 
broadie2011efficient, BauerReussSinger, Cher}. 

\medskip
\noindent
{\bf Financial Market Model} 
We assume there exists a financial market consisting of a money market account 
offering zero interest rate and $d$ risky securities with risk-neutral dynamics\footnote{We are considering 
a standard multi-dimensional Black--Scholes model with zero interest rate for ease of presentation.
One could also use a more complicated financial market model as long as it is possible to efficiently 
simulate from it.
}
\be \label{St}
S^i_t = S^i_0 \exp \brak{\sigma_i B^i_t - \frac{1}{2} \sigma^2_i  t},
\quad t \ge 0,
\ee
for initial prices $S^i_0 = 10$, volatilities $\sigma_i = (10+i/2)\%$ and 
Brownian motions $B^i$, $i = 1, \dots, d$, with instantaneous 
correlation $\rho = 30\%$ between them. We denote the current
time by $0$ and consider a financial derivative on $S^1, \dots, S^d$
with payoff $\phi(S_T)$ at maturity $T = 1/3$ (four months) for a payoff function 
$\phi \colon \R^d \to \R$. Suppose we are interested in the value of the derivative 
at time $t = 1/52$ (one week from now) conditional on the prices 
$S^1_t, \dots, S^d_t$. According to standard no-arbitrage arguments 
\citep[see, e.g.,][]{KS}, it is given by $\E \edg{\phi(S_T) \mid S_t}$, 
which can be written as $\E[Y \mid X]$ for $Y = \phi(S_T)$ and $X_i = S^i_t, \; i = 1, \dots, d$. 
Note that $Y$ has an explicit representation of the form (R) (see the beginning of Section \ref{sec:est})
since $Y$ can be written as $Y= h(X,V)$ for 
\[
h(x,v) = \phi \brak{x_1 \exp \crl{v_1 - \frac{\sigma^2_1}{2} (T-t)}, 
\dots, x_d \exp \crl{v_d - \frac{\sigma^2_d}{2} (T-t)}}\]
and the random variables
\[
V_i = \sigma_i (B^i_T - B^i_t), \quad i = 1, \dots, d,
\]
which are independent of $X_1, \dots, X_d$. 

\medskip
Let us first consider a $d=100$-dimensional max-call option with a time-$T$ payoff of the form
\be \label{max-call}
\phi(S_T) = \brak{\max_{1 \le i \le d} S^i_T - K}^+ 
\ee
with strike price\footnote{The strike price 16.3 has been chosen so that approximately 
half of the simulated paths end up in the money at time $T$.} $K=16.3$. Since the time-$t$ price 
\[
\E \edg{\brak{\max_{1 \le i \le d} S^i_T - K}^+ \, \bigg| \, S_t}
\]
does not admit a closed form solution, it has to be computed numerically. In this example,
\[
h(X, 0) = \brak{\max_{1 \le i \le d} S^i_t e^{- \sigma^2_i (T-t)/2} - K}^+
\]
is zero with high probability. Therefore, it is not useful as an additional feature.
Instead, we use the additional feature 
\be \label{admax-call}
a(X) = \max_{1 \le i \le d} X_i = \max_{1 \le i \le d} S^i_{t}.
\ee The numerical results are reported in 
Table \ref{tab:4}. Additional results are given in Table \ref{tab:10} in the Appendix.   
It can be seen that the three neural networks outperform the second order polynomial 
regression, which works better than the linear regression. The additional 
feature \eqref{admax-call} does not improve the results of any of the methods significantly.

\begin{table}[h!]
\centering
{\footnotesize
\begin{tabular}{l|c|c|c|c|c}
\thead{}  &  95\% CI  $U^{\nu_X}$ &  95\% CI $D^{\nu_X}$ & 
\(\displaystyle \frac{\n{\hat{f} - \bar{f}}_{L^2(\nu_X)}}{\n{\bar{f}}_{L^2(\nu_X)}} \)
& \( \displaystyle \frac{\mbox{\footnotesize 95\% CB } \n{\hat{f} - \bar{f}}_{L^2(\nu_X)}}{\n{\bar{f}}_{L^2(\nu_X)}} \) & \begin{minipage}{9mm} comp. time\\ for $\hat{f}$ \end{minipage} \cr
\midrule
\makecell{lin. regr.}   &  [6.39829, 6.40817] &  [6.39167, 6.40396] &                 4.52 \% &              5.06 \%   &      1.6 s \cr
\myrowcolour
\makecell{lin. regr., add. feature}  &  [6.39771, 6.40759] &  [6.39167, 6.40396] &                 4.27 \% &              4.84 \%    &       1.7 s  \cr
\makecell{poly. regr.}     &  [6.39556, 6.40542] &  [6.39167, 6.40396] &                 3.22 \% &              3.94 \%    &     83.3 s   \cr
\myrowcolour
\makecell{poly. regr., add. feature}   &  [6.39543, 6.40534] &  [6.39167, 6.40396] &                 3.14 \% &              3.88 \%  &      89.86 s   \cr
\makecell{NN tanh}    &  [6.39249, 6.40237] &  [6.39167, 6.40396] &                -1.01 \% &               2.04 \% &       1636.3 s  \cr
\myrowcolour
\makecell{NN tanh, add. feature}   &  [6.39255, 6.40242] &  [6.39167, 6.40396] &                -0.91 \% &               2.09 \%  &      1647.5 s   \cr
\makecell{NN ReLU}    &  [6.39288, 6.40276] &  [6.39167, 6.40396] &  0.60 \% &              2.36 \%   &     1620.1 s    \cr
\myrowcolour
\makecell{NN ReLU, add. feature}   &   [6.39249, 6.40237] &  [6.39167, 6.40396] &                -1.03 \% &              2.04 \%  &      1650.5 s   \cr
\makecell{NN LSE}    &  [6.39276, 6.40263] &  [6.39167, 6.40396] &                 -0.30 \% &              2.26 \%  &     1807.3 s    \cr
\myrowcolour
\makecell{NN LSE, add. feature}     & [6.39260, 6.40247] &  [6.39167, 6.40396] &                -0.83 \% &              2.13 \%  &      1830.9 s  \cr
\bottomrule
\end{tabular}
}
\caption{Numerical results for the max-call option \eqref{max-call}}
\label{tab:4}
\end{table}

\subsection{Binary options}
\label{ex:binary}

In our next example we consider a $d=100$-dimensional binary option in 
the Financial Market Model of Section \ref{ex:max-call} with time-$T$ payoff 
\be \label{bin}
\phi(S_T) = 10 \times 1_{\crl{\max_{1 \le i \le d} S^i_T \ge K}},
\ee
where, as above, we choose $K=16.3$. Again, the time-$t$ price 
\[
10 \, \E \edg{1_{\crl{\max_{1 \le i \le d} S^i_T \ge K}} \, \Big| \,  S_t}
= 10 \, \p \edg{\max_{1 \le i \le d} S^i_T \ge K \, \Big| \, S_t}
\]
cannot be computed exactly and therefore, has to be evaluated numerically. 
As in Section \ref{ex:max-call}, we use
\be \label{adbin} 
a(X) = \max_{1 \le i \le d} X_i = \max_{1 \le i \le d} S^i_t
\ee
as additional feature.

\subsubsection{Minimizing the mean squared distance under $\p$}
\label{ex:BinBP}

We first compute a function $f$ minimizing the mean
squared distance $\E [\brak{Y - f(X)}^2]$ under the original measure $\p$.
Our main numerical results are listed in Table \ref{tab:5}. Additional results 
are given in Table \ref{tab:11} in the Appendix. Again, the three neural networks 
work better than the second order polynomial regression, which is more accurate 
than the linear regression. Adding the additional feature \eqref{adbin} does not 
have a significant influence on any of the methods.

\begin{table}[h!]
\centering
{\footnotesize
\begin{tabular}{l|c|c|c|c|c}
\thead{}  &  95\% CI  $U^{\nu_X}$ &  95\% CI $D^{\nu_X}$ & 
\(\displaystyle \frac{\n{\hat{f} - \bar{f}}_{L^2(\nu_X)}}{\n{\bar{f}}_{L^2(\nu_X)}} \)
& \( \displaystyle \frac{\mbox{\footnotesize 95\% CB } \n{\hat{f} - \bar{f}}_{L^2(\nu_X)}}{\n{\bar{f}}_{L^2(\nu_X)}} \) & \begin{minipage}{9mm} comp. time\\ for $\hat{f}$ \end{minipage} \cr
\midrule
\makecell{lin. regr.}   &  [24.36374, 24.36761] &  [24.33863, 24.36034] &                 2.52 \% &              2.90 \%   &         4.8 s \cr
\myrowcolour
\makecell{lin. regr., add. feature}   &  [24.36276, 24.36666] &  [24.33863, 24.36034] &  2.45 \% &              2.84 \%   &       4.8 s \cr
\makecell{poly. regr.}    &  [24.35346, 24.35739] &  [24.33863, 24.36034] &                 1.56 \% &              2.12 \%     &      88.7 s \cr
\myrowcolour
\makecell{poly. regr., add. feature}    &  [24.35187, 24.35583] &  [24.33863, 24.36034] &  1.37 \% &              1.98 \%  &     93.3 s    \cr
\makecell{NN tanh}    &  [24.34708, 24.35103] &  [24.33863, 24.36034] &    -0.12 \% &   1.43 \%   &         1622.1 s \cr
\myrowcolour
\makecell{NN tanh, add. feature}    &  [24.34615, 24.35011] &  [24.33863, 24.36034] &   -0.58 \% &              1.31 \%  &       1642.5 s   \cr
\makecell{NN ReLU}    &  [24.34672, 24.35067] &  [24.33863, 24.36034] &   -0.38 \% &  1.38 \%  &         1620.2 s \cr
\myrowcolour
\makecell{NN ReLU, add. feature}    &  [24.34605, 24.35001] &  [24.33863, 24.36034] &   -0.61 \% &              1.29 \%  &      1634.3 s    \cr
\makecell{NN LSE}    &   [24.34774, 24.35169] &  [24.33863, 24.36034] &   0.48 \% &     1.51 \%   &         1837.4 s \cr
\myrowcolour
\makecell{NN LSE, add. feature}      & [24.34626, 24.35022] &  [24.33863, 24.36034] &  -0.55 \% &              1.32 \%   &      1832.5 s    \cr
\bottomrule
\end{tabular}
}
\caption{Numerical results for the binary option \eqref{bin} regressed under $\p$}
\label{tab:5}
\end{table}

\subsubsection{Minimizing the mean squared distance under a distorted measure $\p^{\nu}$}
\label{ex:Bindistorted}

In financial risk management, one usually is interested in the tail of a loss distribution.
If a financial institution sold a contract promising a contingent payoff of 
$\phi(S_T)$ at time $T > 0$, the resulting exposure at time $t < T$ is 
$\E[Y \mid X] = \E[\phi(S_T) \mid S_t]$. To obtain a better approximation of 
$\bar{f} = \E[Y \mid X]$ with $\hat{f}(X)$ in the right tail, the least squares regression 
\eqref{empD} can be performed under a measure $\p^{\nu}$ assigning more 
weight to the right tail of $\bar{f}(X)$ than $\p$. This can be done
as in \citet{Cher}. By \eqref{St},
$X = S_t$ can be written as $X = u(QW)$ for a $d$-dimensional standard 
normal random vector $W$, a $d \times d$-matrix $Q$ satisfying 
\[
Q Q^T =  \left( \begin{array}{ccccc}
1 & 0.3 & \dots & \dots & 0.3\\
0.3 & 1 & \dots & \dots & 0.3\\
\dots & \dots & \dots & \dots & \dots\\
0.3 & \dots & \dots & 1 & 0.3\\
0.3 & \dots & \dots & 0.3 & 1
\end{array} 
\right)
\]
and the function $u \colon \R^d \to \R^d$ given by 
\[
u(x_1, \dots, x_d) = 
\brak{S^1_0 \exp \brak{\sigma_1 \sqrt{t} \, x_1 - \frac{1}{2} \sigma^2_1  t}, 
\dots, S^d_0 \exp \brak{\sigma_d \sqrt{t} \, x_d - \frac{1}{2} \sigma^2_d  t}}.
\]
Even though the regression function $\bar{f}$ is not known in closed form, it follows by monotonicity
from the form of the payoff \eqref{bin} that the mapping $\bar{f} \circ u \colon \R^d \to \R$ 
is increasing in the direction $v = (1, \dots, 1)^T$. So, $\bar{f}(X)$ tends to be large
if $v^T Q W$ is large. Therefore, we tilt $\p$ so that the distribution of $W$ shifts in 
the direction of $Q^T v$. 
Let us denote by $q_{\alpha}$ the standard normal quantile at level $\alpha \in (0,1)$.
Since $v^T Q W/\|v^T Q\|_2$ is one-dimensional standard normal, $W$ lies in the region
\[
G = \crl{w \in \R^d : \frac{v^T Q w}{\|v^T Q\|_2} \ge q_{\alpha}}
\]
with probability $1 - \alpha$, whereas, for  
\[
b = \frac{Q^T v}{\|Q^T v\|_2} q_{\alpha} \in \R^d,
\]
$W + b$ lies in $G$ with probability $1/2$. So if $\nu$ is the distribution of 
$u(Q(W+b))$, the $\p^{\nu}$-probability that $\bar{f}(X)$ is in its right 
$\p$-$(1-\alpha)$-tail is approximately 1/2.
Table \ref{tab:6} shows results for the approximation of $\bar{f}$ under $\nu$
corresponding to $\alpha = 0.99$. More details are given in Table \ref{tab:12} in the Appendix.
Again, the three neural networks outperform the polynomial regression, which 
is more accurate than the linear regression, and the inclusion of the
additional feature \eqref{adbin} does not improve the performance of any of 
the methods significantly.

\begin{table}[h!]
\centering
{\footnotesize
\begin{tabular}{l|c|c|c|c|c}
\thead{}  &  95\% CI  $U^{\nu}$ &  95\% CI $D^{\nu}$ & 
\(\displaystyle \frac{\n{\hat{f} - \bar{f}}_{L^2(\nu)}}{\n{\bar{f}}_{L^2(\nu)}} \)
& \( \displaystyle \frac{\mbox{\footnotesize 95\% CB } \n{\hat{f} - \bar{f}}_{L^2(\nu)}}{\n{\bar{f}}_{L^2(\nu)}} \) & \begin{minipage}{9mm} comp. time\\ for $\hat{f}$ \end{minipage} \cr
\midrule
\makecell{lin. regr.}   &   [21.19898, 21.20766] &  [21.1726, 21.19328] & 2.14 \% & 2.36 \%  &  4.9 s \cr
\myrowcolour
\makecell{lin. regr., add. feature}   &  [21.19584, 21.20468] &  [21.1726, 21.19328] & 1.98 \% &              2.21 \%   & 5.2 s \cr
\makecell{poly. regr.}    & [21.19190, 21.20076] &  [21.1726, 21.19328] & 1.76 \% & 2.02 \%  &      
89.1 s \cr
\myrowcolour
\makecell{poly. regr., add. feature}    & [21.18934, 21.19820] &  [21.1726, 21.19328] & 1.60 \% &               1.88 \% &  91.5 s    \cr
\makecell{NN tanh}    & [21.18205, 21.19086] &  [21.1726, 21.19328] & 0.99 \% & 1.40 \% & 1622.1 s \cr
\myrowcolour
\makecell{NN tanh, add. feature}    &  [21.18185, 21.19064] &  [21.1726, 21.19328] & 0.97 \% &              1.38 \%   &  1642.5 s   \cr
\makecell{NN ReLU}    &  [21.18194, 21.19075] &  [21.1726, 21.19328] & 0.98 \% & 1.39 \% & 1620.2 s \cr
\myrowcolour
\makecell{NN ReLU, add. feature}    & [21.18186, 21.19065] &  [21.1726, 21.19328] & 0.97 \% & 
1.39 \% & 1634.3 s    \cr
\makecell{NN LSE}    & [21.18296, 21.19177] &  [21.1726, 21.19328] & 1.09 \% & 1.47 \%  & 1825.2 s \cr
\myrowcolour
\makecell{NN LSE, add. feature} & [21.18185, 21.19065] &  [21.1726, 21.19328] & 0.97 \% & 1.38 \% &    1837.4 s    \cr
\bottomrule
\end{tabular}
}
\caption{Numerical results for the binary option \eqref{bin} regressed under $\p^{\nu}$}
\label{tab:6}
\end{table}

\section{Conclusion}
\label{sec:conclusion}

In this paper, we have studied the numerical approximation of the conditional expectation
of a square-integrable random variable $Y$ given a number of explanatory random variables 
$X_1, \dots, X_d$ by minimizing the mean squared distance between $Y$ and 
$f(X_1, \dots, X_d)$ over a family ${\cal S}$ of Borel functions $f \colon \R^d \to \R$.
The accuracy of the approximation depends on the suitability of the function family ${\cal S}$ 
and the performance of the numerical method used to solve the minimization problem. 
Using an expected value representation of the minimal mean squared 
distance which does not involve a minimization problem or require knowledge of the 
true regression function, we have derived $L^2$-bounds for the approximation error 
of a numerical solution to a given least squares regression problem. We have illustrated the 
method by computing approximations of conditional expectations in a range of examples 
using linear regression, polynomial regression as well as different 
neural network regressions and estimating their $L^2$-approximation errors.
Our results contribute to trustworthy AI by providing numerical guarantees for a
computational problem lying at the heart of different applications
in various fields.

\begin{appendix}
\section{Additional numerical results}

In this appendix we report in Tables \ref{tab:7}--\ref{tab:12}
for all numerical experiments of Section \ref{sec:ex} our estimates 
\begin{itemize}
\item
$U^{\nu}_N$ of the upper bound $U^{\nu} = \E^{\nu}[(Y - \hat{f}(X))^2]$, see 
\eqref{UnuN} and \eqref{Unu},
\item
$D^{\nu}_N$ of the minimal mean squared distance $D^{\nu} = \E^{\nu}[(Y - \bar{f}(X))^2]$,
see \eqref{CDnuN} and \eqref{Dnu},
\item
$F^{\nu}_N$ of the
the squared $L^2$-approximation error $F^{\nu} = \n{\hat{f} - \bar{f}}^2_{L^2(\nu)}$,
see \eqref{FnuN} and  \eqref{Fnu},
\item
$C^{\nu}_N$ of
the squared $L^2$-norm $C^{\nu} = \n{\bar{f}}^2_{L^2(\nu)}$, see \eqref{CDnuN} and \eqref{Cnu},
\end{itemize}
together with the corresponding sample standard errors $\sqrt{v_N^{U,\nu}/N}$, 
$\sqrt{v_N^{D,\nu}/N}$, $\sqrt{v_N^{F,\nu}/N}$ and $\sqrt{v_N^{C,\nu}/N}$, which 
were used to compute\footnote{To fit the numbers reported in this 
appendix into the tables, we had to round them. The numerical results 
in Section \ref{sec:ex} were computed from slightly more precise approximations
of $U^{\nu}_N$, $D^{\nu}_N$, $F^{\nu}_N$, $C^{\nu}_N$ and their 
standard errors.} 
the quantities in Tables \ref{tab:1}--\ref{tab:6} in Section \ref{sec:ex}.

\begin{table}[H]
\centering
{\footnotesize
\begin{tabular}{l|c|c|c|c|c|c|c|c}
\thead{}  & $U_N^{\nu_X}$ &  \( \displaystyle \sqrt{\frac{v^{U,\nu_X}_{N}}{N}}\) & 
\( \displaystyle D_N^{\nu_X}\) & \( \displaystyle\sqrt{\frac{v^{D,\nu_X}_{N}}{N}}\) & 
$F_N^{\nu_X}$ & \( \displaystyle \sqrt{\frac{v^{F,\nu_X}_{N}}{N}}\) & $C_N^{\nu_X}$ 
& \( \displaystyle \sqrt{\frac{v^{C,\nu_X}_{N}}{N}}\)  \cr
\midrule
\makecell{lin. regr.}  
 & 4.00031 &        0.00038 &      1.00011 &        0.00015 &      3.00033 &        0.00023 &      5.00014 &        0.00049 \cr
 \myrowcolour
\makecell{poly. regr.}   
 & 0.99991 &        0.00006 &      1.00011 &        0.00015 &      0.00003 &        0.00012 &      5.00014 &        0.00049 \cr
\makecell{NN tanh}   
 & 0.99998 &        0.00006 &      1.00011 &        0.00015 &      0.00010 &        0.00012 &      5.00014 &        0.00049  \cr
 \myrowcolour
\makecell{NN ReLU}     & 1.00019 &        0.00006 &      1.00011 &        0.00015 &      0.00031 &        0.00012 &      5.00014 &        0.00049 \cr
\makecell{NN LSE}   
 & 0.99999 &        0.00006 &      1.00011 &        0.00015 &      0.00011 &        0.00012 &      5.00014 &        0.00049 \cr
\bottomrule
\end{tabular}
}
\caption{Additional numerical results for the polynomial example \eqref{Ypoly}}
\label{tab:7}
\end{table}

\vspace*{10mm}

\begin{table}[h!]
\centering
{\footnotesize
\begin{tabular}{l|c|c|c|c|c|c|c|c}
{} & $U_N^{\nu_X}$ & \( \displaystyle \sqrt{\frac{v^{U,\nu_X}_{N}}{N}}\) & 
$D_N^{\nu_X}$ & 
\( \displaystyle \sqrt{\frac{v^{D,\nu_X}_{N}}{N}}\)  & $F_N^{\nu_X}$ &  
\( \displaystyle \sqrt{\frac{v^{F,\nu_X}_{N}}{N}}\) & $C_N^{\nu_X}$ & 
\( \displaystyle \sqrt{\frac{v^{C,\nu_X}_{N}}{N}} \)  \cr
\midrule
\makecell{lin. regr.}  
 & 41.57801 &        0.00210 &     36.17079 &        0.00262 &      5.40722 &        0.00171 &      5.40705 &        0.00187 \cr
 \myrowcolour
\makecell{lin. regr., add. feature}  
 & 37.89611 &        0.00214 &     36.17079 &        0.00262 &      1.72471 &        0.00172 &      5.40705 &        0.00187 \cr
\makecell{poly. regr.}  
 &  36.67335 &        0.00202 &     36.17079 &        0.00262 &      0.50198 &        0.00168 &      5.40705 &        0.00187\cr
 \myrowcolour
\makecell{poly. regr., add. feature}    
 & 36.40040 &        0.00205 &     36.17079 &        0.00262 &      0.22893 &        0.00169 &      5.40705 &        0.00187  \cr
\makecell{NN tanh}   
 & 36.17214 &        0.00205 &     36.17079 &        0.00262 &      0.00044 &        0.00169 &      5.40705 &        0.00187  \cr
 \myrowcolour
\makecell{NN tanh, add. feature}    
 & 36.17204 &        0.00205 &     36.17079 &        0.00262 &      0.00030 &        0.00169 &      5.40705 &        0.00187 \cr
\makecell{NN ReLU}   
 & 36.17202 &        0.00205 &     36.17079 &        0.00262 &      0.00029 &        0.00169 &      5.40705 &        0.00187 \cr
 \myrowcolour
\makecell{NN ReLU, add. feature}    
 & 36.17177 &        0.00205 &     36.17079 &        0.00262 &      0.00001 &        0.00169 &      5.40705 &        0.00187 \cr
\makecell{NN LSE}   
 & 36.17160 &        0.00205 &     36.17079 &        0.00262 &     -0.00013 &        0.00169 &      5.40705 &        0.00187  \cr
 \myrowcolour
\makecell{NN LSE, add. feature}      
 & 36.17166 &        0.00205 &     36.17079 &        0.00262 &     -0.00007 &        0.00169 &      5.40705 &        0.00187 \cr
\bottomrule
\end{tabular}
}
\caption{Additional numerical results for the non-polynomial example
\eqref{non-poly} regressed under $\p$}
\label{tab:8}
\end{table}

\newpage

\begin{table}[H]
\centering
{\footnotesize
\begin{tabular}{l|c|c|c|c|c|c|c|c}
{} & $U_N^{\nu}$ &  \( \displaystyle \sqrt{\frac{v^{U,\nu}_{N}}{N}}\) & $D_N^{\nu}$ &  
\( \displaystyle \sqrt{\frac{v^{D,\nu}_{N}}{N}}\)  & $F_N^{\nu}$ & 
\( \displaystyle \sqrt{\frac{v^{F,\nu}_{N}}{N}} \) & $C_N^{\nu}$ &  
\( \displaystyle \sqrt{\frac{v^{C,\nu}_{N}}{N}}\)  \cr
\midrule
\makecell{lin. regr.}  
 & 39.93609 &        0.00212 &     39.84922 &        0.00271 &      0.08475 &        0.00169 &     11.05717 &        0.00209  \cr
 \myrowcolour
\makecell{lin. regr., add. feature}  
 & 39.92639 &        0.00212 &     39.84922 &        0.00271 &      0.07508 &        0.00169 &     11.05717 &        0.00209 \cr
\makecell{poly. regr.}   
 & 39.85559 &        0.00212 &     39.84922 &        0.00271 &      0.00446 &        0.00169 &     11.05717 &        0.00209 \cr
 \myrowcolour
\makecell{poly. regr., add. feature}   
 &  39.85515 &        0.00211 &     39.84922 &        0.00271 &      0.00406 &        0.00169 &     11.05717 &        0.00209 \cr
\makecell{NN tanh}   
 & 39.85126 &        0.00212 &     39.84922 &        0.00271 &      0.00012 &        0.00169 &     11.05717 &        0.00209 \cr
 \myrowcolour
\makecell{NN tanh, add. feature}     
 & 39.85131 &        0.00212 &     39.84922 &        0.00271 &      0.00018 &        0.00169 &     11.05717 &        0.00209 \cr
\makecell{NN ReLU}   
 & 39.85132 &        0.00212 &     39.84922 &        0.00271 &      0.00018 &        0.00169 &     11.05717 &        0.00209 \cr
 \myrowcolour
\makecell{NN ReLU, add. feature}    
 &  39.85139 &        0.00212 &     39.84922 &        0.00271 &      0.00026 &        0.00169 &     11.05717 &        0.00209  \cr
\makecell{NN LSE}   
 &  39.85125 &        0.00212 &     39.84922 &        0.00271 &      0.00012 &        0.00169 &     11.05717 &        0.00209 \cr
 \myrowcolour
\makecell{NN LSE, add. feature}      
 &  39.85132 &        0.00212 &     39.84922 &        0.00271 &      0.00018 &        0.00169 &     11.05717 &        0.00209  \cr
\bottomrule
\end{tabular}
}
\caption{Additional numerical results for the non-polynomial example regressed under
$\p^{\nu}$}
\label{tab:9}
\end{table}

\vspace*{12mm}

\begin{table}[h!]
\centering
{\footnotesize
\begin{tabular}{l|c|c|c|c|c|c|c|c}
{}  & $U_N^{\nu_X}$ &  \( \displaystyle \sqrt{\frac{v^{U,\nu_X}_{N}}{N}}\) & $D_N^{\nu_X}$ &  
\( \displaystyle \sqrt{\frac{v^{D,\nu_X}_{N}}{N}} \) & $F_N^{\nu_X}$ &  
\( \displaystyle \sqrt{\frac{v^{F,\nu_X}_{N}}{N}} \) & $C_N^{\nu_X}$ &  
\( \displaystyle \sqrt{\frac{v^{C,\nu_X}_{N}}{N}} \)  \cr
\midrule
\makecell{lin. regr.}  
 & 6.40323 &        0.00252 &      6.39782 &        0.00313 &      0.00552 &        0.00086 &      2.70728 &        0.00119 \cr
 \myrowcolour
\makecell{lin. regr., add. feature}  
 & 6.40265 &        0.00252 &      6.39782 &        0.00313 &      0.00494 &        0.00086 &      2.70728 &        0.00119 \cr
\makecell{poly. regr.}  
 & 6.40049 &        0.00252 &      6.39782 &        0.00313 &      0.00280 &        0.00086 &      2.70728 &        0.00119  \cr
 \myrowcolour
\makecell{poly. regr., add. feature}  
 & 6.40038 &        0.00253 &      6.39782 &        0.00313 &      0.00267 &        0.00086 &      2.70728 &        0.00119 \cr
\makecell{NN tanh}   
 & 6.39743 &        0.00252 &      6.39782 &        0.00313 &     -0.00028 &        0.00086 &      2.70728 &        0.00119 \cr
 \myrowcolour
\makecell{NN tanh, add. feature}    
 &  6.39749 &        0.00252 &      6.39782 &        0.00313 &     -0.00023 &        0.00086 &      2.70728 &        0.00119 \cr
\makecell{NN ReLU}   
 & 6.39782 &        0.00252 &      6.39782 &        0.00313 &      0.00010 &        0.00086 &      2.70728 &        0.00119  \cr
 \myrowcolour
\makecell{NN ReLU, add. feature}    
 &  6.39743 &        0.00252 &      6.39782 &        0.00313 &     -0.00029 &        0.00086 &      2.70728 &        0.00119 \cr
\makecell{NN LSE}   
 & 6.39769 &        0.00252 &      6.39782 &        0.00313 &     -0.00002 &        0.00086 &      2.70728 &        0.00119  \cr
 \myrowcolour
\makecell{NN LSE, add. feature}      
 & 6.39754 &        0.00252 &      6.39782 &        0.00313 &     -0.00018 &        0.00086 &      2.70728 &        0.00119 \cr
\bottomrule
\end{tabular}
}
\caption{Additional numerical results for the max-call option \eqref{max-call}}
\label{tab:10}
\end{table}

\newpage

\begin{table}[H]
\centering
{\footnotesize
\begin{tabular}{l|c|c|c|c|c|c|c|c}
{} & $U_N^{\nu_X}$ & \( \displaystyle \sqrt{\frac{v^{U,\nu_X}_{N}}{N}} \) 
& $D_N^{\nu_X}$ & \( \displaystyle \sqrt{\frac{v^{D,\nu_X}_{N}}{N}}\)  & $F_N^{\nu_X}$ &  
\( \displaystyle \sqrt{\frac{v^{F,\nu_X}_{N}}{N}} \) & $C_N^{\nu_X}$ 
& \( \displaystyle \sqrt{\frac{v^{C,\nu_X}_{N}}{N}} \)  \cr
\midrule
\makecell{lin. regr.}  
 & 24.36567 &        0.00099 &     24.34948 &        0.00554 &      0.01644 &        0.00322 &     25.87654 &        0.00565 \cr
 \myrowcolour
\makecell{lin. regr., add. feature}  
 & 24.36471 &        0.00099 &     24.34948 &        0.00554 &      0.01553 &        0.00322 &     25.87654 &        0.00565 \cr
\makecell{poly. regr.}   
 & 24.35543 &        0.00100 &     24.34948 &        0.00554 &      0.00631 &        0.00322 &     25.87654 &        0.00565 \cr
 \myrowcolour
\makecell{poly. regr., add. feature}   
 & 24.35385 &        0.00101 &     24.34948 &        0.00554 &      0.00483 &        0.00322 &     25.87654 &        0.00565  \cr
\makecell{NN tanh}  
 &  24.34906 &        0.00101 &     24.34948 &        0.00554 &     -0.00003 &        0.00322 &     25.87654 &        0.00565 \cr
 \myrowcolour
\makecell{NN tanh, add. feature}     
 & 24.34813 &        0.00101 &     24.34948 &        0.00554 &     -0.00088 &        0.00322 &     25.87654 &        0.00565 \cr
\makecell{NN ReLU}   
 & 24.34869 &        0.00101 &     24.34948 &        0.00554 &     -0.00038 &        0.00322 &     25.87654 &        0.00565 \cr
 \myrowcolour
\makecell{NN ReLU, add. feature}    
 & 24.34803 &        0.00101 &     24.34948 &        0.00554 &     -0.00098 &        0.00322 &     25.87654 &        0.00565 \cr
\makecell{NN LSE}   
 & 24.34972 &        0.00101 &     24.34948 &        0.00554 &      0.00059 &        0.00322 &     25.87654 &        0.00565 \cr
 \myrowcolour
\makecell{NN LSE, add. feature}      
 & 24.34824 &        0.00101 &     24.34948 &        0.00554 &     -0.00077 &        0.00322 &     25.87654 &        0.00565  \cr
\bottomrule
\end{tabular}
}
\caption{Additional numerical results for the binary option \eqref{bin}
regressed under $\p$}
\label{tab:11}
\end{table} 

\vspace*{12mm}

\begin{table}[H]
\centering
{\footnotesize
\begin{tabular}{l|c|c|c|c|c|c|c|c}
{} & $U_N^{\nu}$ & \( \displaystyle \sqrt{\frac{v^{U,\nu}_{N}}{N}} \) 
& $D_N^{\nu}$ & \( \displaystyle \sqrt{\frac{v^{D,\nu}_{N}}{N}}\)  & $F_N^{\nu}$ &  
\( \displaystyle \sqrt{\frac{v^{F,\nu}_{N}}{N}} \) & $C_N^{\nu}$ 
& \( \displaystyle \sqrt{\frac{v^{C,\nu}_{N}}{N}} \)  \cr
\midrule
\makecell{lin. regr.}  &
21.20332 &        0.00221 &     21.18294 &        0.00528 &      0.02149 &        0.00279 &     46.90835 &        0.00644 \cr
 \myrowcolour
\makecell{lin. regr., add. feature}  
&     21.20026 &        0.00226 &     21.18294 &        0.00528 &      0.01841 &        0.00279 &     46.90835 &        0.00644  \cr 
\makecell{poly. regr.}  &
21.19633 &        0.00226 &     21.18294 &        0.00528 &      0.01455 &        0.00279 &     46.90835 &        0.00644 \cr
 \myrowcolour
\makecell{poly. regr., add. feature}  &     21.19377 &        0.00226 &     21.18294 &        0.00528 &      0.01198 &        0.00279 &     46.90835 &        0.00644 \cr
\makecell{NN tanh}  &
21.18645 &        0.00225 &     21.18294 &        0.00528 &      0.00466 &        0.00279 &     46.90835 &        0.00644 \cr
 \myrowcolour
\makecell{NN tanh, add. feature} &     21.18624 &        0.00224 &     21.18294 &        0.00528 &      0.00442 &        0.00279 &     46.90835 &        0.00644 \cr
\makecell{NN ReLU}  &
21.18635 &        0.00225 &     21.18294 &        0.00528 &      0.00454 &        0.00279 &     46.90835 &        0.00644 \cr
 \myrowcolour
\makecell{NN ReLU, add. feature} & 21.18625 &        0.00224 &     21.18294 &        0.00528 &      0.00443 &        0.00279 &     46.90835 &        0.00644 \cr
\makecell{NN LSE}  &
21.18736 &        0.00225 &     21.18294 &        0.00528 &      0.00555 &        0.00279 &     46.90835 &        0.00644 \cr
 \myrowcolour
\makecell{NN LSE, add. feature}  &     21.18625 &        0.00224 &     21.18294 &        0.00528 &      0.00442 &        0.00279 &     46.90835 &        0.00644 \cr
\bottomrule
\end{tabular}
}
\caption{Additional numerical results for the binary option \eqref{bin} regressed 
under $\p^{\nu}$}
\label{tab:12}
\end{table} 
\end{appendix}


\newpage
\begin{thebibliography}{99}

\bibitem[Acerbi and Tasche(2002)]{AT}
Carlo Acerbi and Dirk Tasche.
\newblock On the coherence of expected shortfall.
\emph{Journal of Banking \& Finance} 26, 1487--1503, 2002.

\bibitem[{\AA}str\"{o}m(1970)]{astrom70}
Karl~J. {\AA}str\"{o}m.
\newblock \emph{Introduction to Stochastic Control Theory}, Vol. ~70 of
\emph{Mathematics in Science and Engineering}.
\newblock Academic Press, New York-London, 1970.

\bibitem[Bain and Crisan(2008)]{bain2008fundamentals}
Alan Bain and Dan Crisan.
\newblock \emph{{Fundamentals of Stochastic Filtering}}, Vol. ~60.
\newblock Springer Science \& Business Media, 2008.

\bibitem[Bally(1997)]{bally1997approximation} Vlad Bally.
\newblock {Approximation scheme for solutions of BSDE}.
\newblock \emph{Pitman Research Notes in Mathematics Series}, Longman 364, 1997.

\bibitem[Bauer et~al.(2012)Bauer, Reuss, and Singer]{BauerReussSinger}
Daniel Bauer, Andreas Reuss and Daniela Singer.
\newblock On the calculation of the solvency capital requirement based on
  nested simulations.
\newblock \emph{{ASTIN Bulletin}} 42, 453--499, 2012.

\bibitem[Beck et~al.(2021)Beck, Becker, Cheridito, Jentzen and
  Neufeld]{beck2019deep}
Christian Beck, Sebastian Becker, Patrick Cheridito, Arnulf Jentzen, and Ariel
  Neufeld.
\newblock {Deep splitting method for parabolic PDEs}.
\newblock \emph{{SIAM Journal on Scientific Computing}} 43(5), A3135--A3154, 2021.

\bibitem[Beck et~al.(2020)Beck, Becker, Cheridito, Jentzen and
  Neufeld]{beck2020deep}
Christian Beck, Sebastian Becker, Patrick Cheridito, Arnulf Jentzen and Ariel
  Neufeld.
\newblock Deep learning based numerical approximation algorithms for stochastic
  partial differential equations and high-dimensional nonlinear filtering
  problems.
\newblock \emph{arXiv:2012.01194}, 2020.

\bibitem[Becker et~al.(2020)Becker, Cheridito, and Jentzen]{becker2020pricing}
Sebastian Becker, Patrick Cheridito and Arnulf Jentzen.
\newblock {Pricing and hedging American-style options with deep learning}.
\newblock \emph{{Journal of Risk and Financial Management}} 13(7), 158, 1--12, 2020.

\bibitem[Bj\"{o}rck(1996)]{bjorck96}
{\AA}ke Bj\"{o}rck.
\newblock \emph{Numerical Methods for Least Squares Problems}.
\newblock SIAM, Philadelphia, PA, 1996.

\bibitem[Bouchard and Touzi(2004)]{BT}
Bruno Bouchard and Nizard Touzi.
\newblock 
Discrete-time approximation and Monte-Carlo simulation of backward stochastic differential equations.
\newblock
\emph{Stochastic Processes and their Applications} 11(2), 175--206.

\bibitem[Broadie and Cao(2008)]{broadie2008}
Mark Broadie and Menghui Cao.
\newblock Improved lower and upper bound algorithms for pricing {A}merican
  options by simulation.
\newblock \emph{Quant. Finance} 8, 845--861, 2008.

\bibitem[Broadie et~al.(2011)Broadie, Du, and Moallemi]{broadie2011efficient}
Mark Broadie, Yiping Du and Ciamac~C. Moallemi.
\newblock {Efficient risk estimation via nested sequential simulation}.
\newblock \emph{{Management Science}} 57, 1172--1194,
  2011.

\bibitem[Broadie et~al.(2015)Broadie, Du, and Moallemi]{broadie2015risk}
Mark Broadie, Yiping Du and Ciamac~C. Moallemi.
\newblock {Risk estimation via regression}.
\newblock \emph{{Operations Research}} 63, 1077--1097, 2015.

\bibitem[Broadie and Glasserman(2004)]{broadie2004stochastic}
Mark Broadie and Paul Glasserman.
\newblock A stochastic mesh method for pricing high-dimensional American
  options.
\newblock \emph{Journal of Computational Finance} 7, 35--72, 2004.

\bibitem[Bru and Heinich(1985)]{bru1985meilleures}
Bernard Bru and Henri Heinich.
\newblock {Meilleures approximations et m{\'e}dianes conditionnelles}.
\newblock In \emph{Annales de l'IHP Probabilit{\'e}s et Statistiques} 21, 197--224, 1985.

\bibitem[Carriere(1996)]{carriere96}
Jacques~F. Carriere.
\newblock Valuation of the early-exercise price for options using simulations
  and nonparametric regression.
\newblock \emph{Insurance Math. Econom.} 19, 19--30, 1996.

\bibitem[Chatterjee and Hadi(2015)]{chatterjee2015regression}
Samprit Chatterjee and Ali~S. Hadi.
\newblock \emph{{Regression Analysis by Example}}.
\newblock John Wiley \& Sons, 2015.

\bibitem[Cheridito et.~al(2020)]{Cher}
Patrick Cheridito, John Ery and Mario V. Wüthrich.
Assessing asset-liability risk with neural networks.
\newblock {Risks} 8(1), 16, 1--17, 2020.

\bibitem[Chevance(1997)]{chevance1997numerical}
David Chevance.
\newblock {Numerical methods for backward SDEs}.
\newblock \emph{{Numerical Methods in Finance}} 232, 1997.

\bibitem[Draper and Smith(1998)]{norman98}
Norman~R. Draper and Harry Smith.
\newblock \emph{Applied Regression Analysis}.
\newblock John Wiley \& Sons, Inc., New York, 1998.

\bibitem[Fahim et~al.(2011)]{FTW}
Arash Fahim, Nizar Touzi and Xavier Warin.
\newblock {A probabilistic numerical method for fully nonlinear parabolic PDEs}.
\newblock \emph{{The Annals of Applied Probability}} 21(4), 1322--1364, 2011.

\bibitem[F\"ollmer and Schied(2016)]{FS}
Hans Föllmer and Alexander Schied.
\newblock \emph{Stochastic Finance}.
De Gruyter Textbook, 2016.

\bibitem[Gelman et~al.(2013)Gelman, Carlin, Stern, Dunson, Vehtari, and Rubin]{gelman2013bayesian}
Andrew Gelman, John~B. Carlin, Hal~S. Stern, David~B. Dunson, Aki Vehtari and
  Donald~B. Rubin.
\newblock \emph{Bayesian Data Analysis}.
\newblock CRC Press, 2013.

\bibitem[Glorot and Bengio(2010)]{glorot2010understanding}
Xavier Glorot and Yoshua Bengio.
\newblock {Understanding the difficulty of training deep feedforward neural networks}.
\newblock In \emph{Proceedings of the 13th International Conference on
  Artificial Intelligence and Statistics}, 249--256, 2010.

\bibitem[Gobet et~al.(2005)Gobet, Lemor, and Warin]{gobet2005regression}
Emmanuel Gobet, Jean-Philippe Lemor and Xavier Warin.
\newblock {A regression-based {M}onte {C}arlo method to solve backward SDEs}.
\newblock \emph{{The Annals of Applied Probability}} 15, 2172--2202, 2005.

\bibitem[Gobet and Turkedjiev(2006)]{GobetT}
Emmanuel Gobet and Plamen Turkedjiev.
\newblock {Linear regression MDP scheme for discrete backward stochastic differential 
equations under general conditions}.
\newblock \emph{Math. Comp.} 85, 1359--1391, 2006.

\bibitem[Goodfellow et~al.(2016)Goodfellow, Bengio, and Courville]{goodfellow2016deep}
Ian Goodfellow, Yoshua Bengio and Aaron Courville.
\newblock \emph{{Deep Learning}}, Vol. ~1.
\newblock MIT Press Cambridge, 2016.

\bibitem[Gordy and Juneja(2010)]{gordy2010nested}
Michael~B. Gordy and Sandeep Juneja.
\newblock Nested simulation in portfolio risk measurement.
\newblock \emph{Management Science} 56, 1833--1848,  2010.

\bibitem[Hastie et~al.(2009)Hastie, Tibshirani, and
  Friedman]{hastie2009elements}
Trevor Hastie, Robert Tibshirani and Jerome Friedman.
\newblock \emph{{The Elements of Statistical Learning: Data Mining, Inference,
  and Prediction}}.
\newblock Springer Science \& Business Media, 2009.

\bibitem[Ioffe and Szegedy(2015)]{IoffeSzegedy}
Sergey Ioffe and Christian Szegedy.
\newblock Batch normalization: accelerating deep network training by reducing
  internal covariate shift.
\newblock In \emph{Proceedings of the 32nd International Conference on Machine Learning}
37, 448--456, 2015. 

\bibitem[Jazwinski(2007)]{jazwinski2007stochastic}
Andrew~H. Jazwinski.
\newblock \emph{Stochastic Processes and Filtering Theory}.
\newblock Courier Corporation, 2007.

\bibitem[Karatzas and Shreve(2010)]{KS}
Ioannis Karatzas and Steven E. Shreve.
\newblock \emph{Methods of Mathematical Finance}.
Springer, 1998.

\bibitem[Kingma and Ba(2014)]{kingma2014adam}
Diederik~P. Kingma and Jimmy Ba.
\newblock {Adam: A method for stochastic optimization}.
\newblock \emph{{arXiv:1412.6980}}, 2014.

\bibitem[Lee and Glynn(2003)]{lee2003computing}
Shing-Hoi Lee and Peter~W. Glynn.
\newblock Computing the distribution function of a conditional expectation via
  {M}onte {C}arlo: Discrete conditioning spaces.
\newblock \emph{ACM Transactions on Modeling and Computer Simulation}
13, 238--258, 2003.

\bibitem[Longstaff and Schwartz(2001)]{longstaff2001valuing}
Francis~A. Longstaff and Eduardo~S. Schwartz.
\newblock {Valuing American options by simulation: a simple least-squares
  approach}.
\newblock \emph{The Review of Financial Studies} 14, 113--147, 2001.

\bibitem[Ryan(2009)]{ryan2008modern}
Thomas~P. Ryan.
\newblock \emph{Modern regression methods}.
\newblock Wiley Series in Probability and Statistics. John Wiley \& Sons, Inc.,
  Hoboken, NJ, 2nd Edition, 2009.

\bibitem[Tsitsiklis and Van~Roy(2001)]{tsitsiklis2001regression}
John~N. Tsitsiklis and Benjamin Van~Roy.
\newblock {Regression methods for pricing complex American-style options}.
\newblock \emph{{IEEE Transactions on Neural Networks}} 12, 694--703, 2001.

\end{thebibliography}
\end{document}